\documentclass[journal]{IEEEtran} 
\IEEEoverridecommandlockouts
\usepackage{cite}
\usepackage{amsmath,amssymb,amsfonts,amsthm}
\usepackage{algorithmic}
\usepackage{graphicx}
\usepackage{textcomp}
\usepackage{multirow}
\usepackage{xcolor}
\usepackage{bm,color}
\usepackage[citecolor=red,urlcolor=red]{hyperref}
\hypersetup{
	colorlinks=true,
	linkcolor=blue,
	filecolor=red,      
	urlcolor=blue,
	citecolor=blue,
}
\usepackage{graphicx} 
\usepackage{subfigure}
\usepackage{colortbl}

\newtheorem{rem}{Remark}

\usepackage{booktabs}

\begin{document}
%
\title{ Excitation Operator based Fault Separation Applied to a Quadrotor UAV}

%

\author{Sicheng Zhou,
	    Meng Wang,
	    Jindou Jia,
	    Kexin Guo,
	    Xiang Yu,
	    Youmin Zhang,~\IEEEmembership{Fellow,~IEEE}
	   and  Lei Guo*,~\IEEEmembership{Fellow,~IEEE}
\thanks{This work was partially supported by the National Key Research and Development Program of China (2020YFA0711200), National Natural science Foundation of China (Grant No. 61903019, 61973012, 62003018 and 62173024), Key Research and Development Program of Zhejiang(2021C03158),  Zhejiang Provincial Natural Science Foundation (Grant No.LO20F030006 and LD21F030001).}
\thanks{S. Zhou, M. Wang, and J. Jia are with the School of Automation Science and Electrical Engineering, Beihang University, 100191, Beijing, China. Email: \{zb1903003, mwbuaa001, jdjia\}@buaa.edu.cn.}
\thanks{K. Guo is with the School of Aeronautic Science and Engineering, Beihang University, 100191, Beijing, China. Email: kxguo@buaa.edu.cn.}
\thanks{X. Yu and L. Guo are with the School of Automation Science and Electrical Engineering, Beihang University, 100191, Beijing, China, and Hangzhou Innovation Institute, Beihang University, 310051, Hangzhou, China. Email: \{xiangyu\_buaa, lguo\}@buaa.edu.cn.}
\thanks{Y. M. Zhang is with the Department of Mechanical, Industrial and Aerospace Engineering, Concordia University, Montreal, QC H3G 1M8, Canada. Email: ymzhang@encs.concordia.ca.}
\thanks{* Corresponding author.}}

%
%

\markboth{Journal of \LaTeX\ Class Files,~Vol.~, No.~, Feb~2023}%
{Shell \MakeLowercase{\textit{et al.}}: Fault Separation for a Quadrotor UAV based on an Excitation Operator}
%



\maketitle

\begin{abstract}
This paper presents an excitation operator based fault separation architecture for a quadrotor unmanned aerial vehicle subject to actuator faults, actuator aging, and load uncertainty. The actuator fault dynamics is deeply excavated, containing the deep coupling information among the actuator faults, the system states, and control inputs. By explicitly considering the physical constraints and tracking performance, an excitation operator and corresponding integrated state observer are designed to estimate separately actuator fault and load uncertainty. Moreover, a fault separation maneuver and a safety controller are proposed to ensure the tracking performance when the excitation operator is injected. Both comparative simulation and flight experiments have demonstrated the effectiveness of the proposed fault separation scheme while maintaining high levels of tracking performance.
\end{abstract}

\begin{IEEEkeywords}
 Fault separation, excitation operator, safety control, quadrotor UAV.
\end{IEEEkeywords}

%
\IEEEpeerreviewmaketitle

\section{Introduction}
%
%
%
%
\IEEEPARstart{E}{nsuring} system safety during flight is crucial to the development of quadrotor unmanned aerial vehicle (UAV). Fault detection and diagnosis (FDD) plays an essential role in ensuring the safety of quadrotor UAV \cite{yu2015survey,9385897}. However, external disturbances, system uncertainties, and input delays, along with compensatory control maneuver, usually affect the reliability and accuracy of fault diagnosis \cite{guo2014anti,guo2005disturbance}. It is difficult to guarantee the separability of faults due to the coupling relationship of such external disturbances, system uncertainties, and measurement noises.

Tremendous efforts have been devoted to addressing the problem of FDD \cite{9279322,9124667}.In \cite{jia2022novel}, an optimization-based framework for the cooperative design of active FDD and FTC is developed to increase the accuracy of fault diagnosis and maintain tracking performance. In \cite{10113778}, a new health indicator construction method based on a quantitative estimation neural network is proposed and can effectively reduce the impact of variable speeds. An adaptive Kalman filter is proposed for actuator fault diagnosis in the stochastic framework \cite{ZHANG2018333}. In \cite{9992045}, by explicitly considering the real physical constraints, a sliding mode observer (SMO) is implemented on the autopilot of a quadrotor UAV. An adaptive activation transfer learning approach is proposed for industrial scenarios which can achieve fault diagnosis under the uncertainties and measurement noises \cite{10056250}.

However, the above methods all treat the faults as lumped terms for composite estimation, which significantly reduces the accuracy of FDD under disturbances or uncertainties. In \cite{9743552}, a finite-time sliding-mode observer is designed to detect, isolate, and identify actuator faults under the presence of external disturbances. In \cite{cao2017anti}, the authors propose a disturbance observer (DO) and a fault diagnosis observer to estimate the unknown modeled disturbance and control fault, respectively. In a similar vein, An adaptive augmented state/FDD observer is presented to estimate the system state, sensor and actuator faults simultaneously \cite{9200647}. In \cite{8301571}, a fault-tolerant estimation approach is developed by combining sensor FDD results and air data reconstruction. Unlike the previous works, this approach does not explicitly consider external disturbances, but rather focuses on the accurate estimation of sensor faults and air data in real-time.  

However, there still exist several technical challenges to improve the reliability and accuracy of the FDD in the quadrotor safety control system.

\begin{enumerate}

\item Most of the traditional passive fault diagnosis approaches determine the health status of the quadrotor UAV system by monitoring system states. It is difficult to detect system degradation and minor faults. Meanwhile, the false alarm rate and missed detection rate of the FDD module increase. As a result, the separation estimation of the coupled faults and the uncertainties is of paramount importance to improve the robustness and accuracy of FDD.




\item The active fault diagnosis method requires redundant control channels to inject auxiliary input signals. It can be directly applied to the control of over-actuated systems. Nevertheless, it is difficult to apply the active method to the under-actuated quadrotor UAV system. Besides, auxiliary input signal may cause a secondary damage of the quadrotor UAV system.
\end{enumerate}

Motivated by addressing the preceding issues, this paper focuses on handling the actuator faults, aging, and load uncertainty in a separation manner. Similar to the active FDD, an excitation operator, including two kinds of auxiliary input signals, is designed by taking the physical constraints and tracking performance into consideration. A fault separation maneuver of quadrotor UAV is designed to guarantee the under-actuated system safety when the auxiliary input signals are injected. An integrated state observer can achieve the separate estimation of actuator faults, aging, and uncertainties according to the characteristics of auxiliary input signals. The technical contributions are summarized as follows:

\begin{enumerate}
\item	The relationships among the actuator fault dynamics, control input, system state, and load uncertainty are revealed in this study. In comparison of the existing work \cite{9684676}, the dynamic of actuator fault has been fully used, reducing the conservativeness.

\item	 An excitation operator is designed to decrease the impact of actuator faults and aging on the system. Meanwhile, the physical constraints and tracking performance of the system can be guaranteed. When comparing to the passive FDD methods \cite{freddi2012diagnostic, cao2017anti, cen2014robust}, the auxiliary input signals can assist the separation among actuator fault, aging, and load uncertainty.

\item In comparison of the active FDD algorithm \cite{jia2022novel}, the condition of redundant channels is no longer required because of the proposed fault separation maneuver. By sacrificing the controllability of yaw channel, a redundant control channel can be created to inject auxiliary input signals. Meanwhile, the secondary damage caused by excitation operator can be avoided.
\end{enumerate}

The remainder of this paper includes four sections. In Section II, the mathematical model and actuator fault deep-coupling model of a quadrotor UAV are proposed. The fault separation scheme is presented in Section III, including the auxiliary input signals design,  integrated FDD module, and the corresponding safety control scheme. Section IV uses numerical simulations and flight tests to illustrate the effectiveness of the proposed scheme in this paper. Finally, this paper is concluded in Section V.

\section{Mathematical Model of Quadrotor UAV}
In order to validate the proposed method, the mathematical quadrotor model and the actuator fault deep-coupling model are established firstly. As shown in Fig. \ref{fig_1}, two coordinate reference frames are usually involved. The Body Frame (BF) $\mathcal{F}\left({O}_{B}{x}_{B}{y}_{B}{z}_{B}\right)$ is the body-fixed frame at the center of gravity (CoG) of the quadrotor UAV. The Inertial Frame (IF) $\mathcal{F}\left({O}_{E}{x}_{E}{y}_{E}{z}_{E}\right)$ is an Earth-fixed inertial frame at a defined location ${O}_{E}$. 


\begin{figure}[!t]
 \centering
 \includegraphics[width=2.5in]{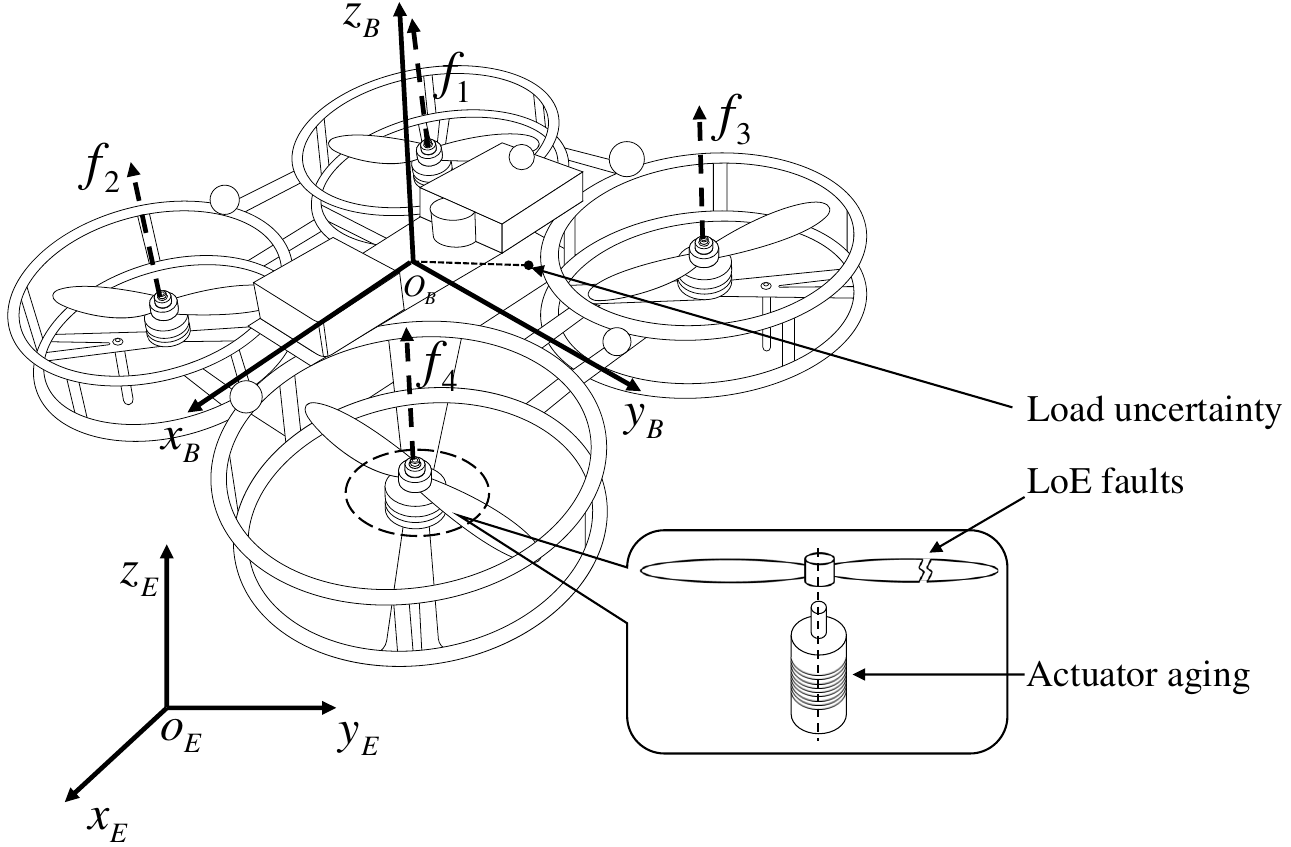}
 \caption{The structure of a quadrotor UAV with the body frame and inertial frame.}
 \label{fig_1}
\end{figure}

The coordinate transform matrix from the BF to the IF can be expressed as:
\begin{equation}
\label{eq_1}
{{R}_{EB}}=\left[ \begin{matrix}
{{{C}_{\psi }}}{{C}_{\theta }} & -{{S}_{\psi }}{{C}_{\phi }}+{{{C}_{\psi }}}{{S}_{\theta }}{{S}_{\phi }} & {{S}_{\psi }}{{S}_{\phi }}+{{{C}_{\psi }}}{{S}_{\theta }}{{C}_{\phi }}  \\
{{S}_{\psi }}{{C}_{\theta }} & {{{C}_{\psi }}}{{C}_{\phi }}+{{S}_{\psi }}{{S}_{\theta }}{{S}_{\phi }} & -{{{C}_{\psi }}}{{S}_{\phi }}+{{S}_{\psi }}{{S}_{\theta }}{{C}_{\phi }}  \\
-{{S}_{\theta }} & {{C}_{\theta }}{{S}_{\phi }} & {{C}_{\theta }}{{C}_{\phi }}  \\
\end{matrix} \right],
\end{equation}
where ${{\Omega }_{b}}={{[\begin{matrix}
		\phi  & \theta  & \psi   \\
		\end{matrix}]}^{T}}$ represents body-axis pitch, roll, and yaw angle. {${{S}_{*}}$ and ${{C}_{*}}$ denote $\sin \left( *  \right)$ and $\cos \left( *  \right)$}, respectively.

\subsection{Nonlinear Quadrotor UAV Model}
The quadrotor UAV is regarded as a rigid body with constant mass $m$ and constant moment of inertia $J$. The nonlinear dynamic equations of a quadrotor UAV can be modeled as:
\begin{equation}
\label{eq_2}
\left\{ \begin{array}{*{35}{l}}
{{{\dot{p}}}_{E}}={{v}_{E}},  \\
{{{\dot{v}}}_{E}}=\frac{1}{m}{{R}_{EB}}\left( {{\Omega }_{b}} \right){{F}_{m}}+G,  \\
{{{\dot{\Omega }}}_{b}}={{R}_{{0}}}\left( {{\Omega }_{b}} \right)\omega,   \\
\dot{\omega }={\tilde{D}}\left( {{\Omega }_{b}},\omega  \right)+{{D}_{M}}+\Theta \left( {{J}^{-1}_{x}},{{J}^{-1}_{y}},{{J}^{-1}_{z}} \right){{\xi }_{\omega }},  \\
\end{array} \right.
\end{equation}
where $G={{\left[ \begin{matrix}
		0 & 0 & -g  \\
		\end{matrix} \right]}^{T}}$ with $g$ is the gravity constant,
\begin{equation}
\label{eq_3}
{{R}_{{0}}}\left( {{\Omega }_{b}} \right)=\left[ \begin{matrix}
1 & {{S}_{\theta }}C_{\theta }^{-1}{{S}_{\phi }} & {{S}_{\theta }}C_{\theta }^{-1}{{C}_{\phi }}  \\
0 & {{C}_{\phi }} & -{{S}_{\phi }}  \\
0 & C_{\theta }^{-1}{{S}_{\phi }} & C_{\theta }^{-1}{{C}_{\phi }}  \\
\end{matrix} \right],
\end{equation}

\begin{equation}
\label{eq_4}
{\tilde{D}}\left( {{\Omega }_{b}},\omega  \right)={{\left[ \begin{matrix}
		\frac{{{J}_{y}}-{{J}_{z}}}{{{J}_{x}}}qr & \frac{{{J}_{z}}-{{J}_{x}}}{{{J}_{y}}}pr & \frac{{{J}_{x}}-{{J}_{y}}}{{{J}_{z}}}pq  \\
		\end{matrix} \right]}^{T}}
\end{equation}
are the state vector of the rotation, and
\begin{equation}
\label{eq_5}
{{D}_{M}}={{\left[ \begin{matrix}
		J_{x}^{-1}{{M}_{x}} & J_{y}^{-1}{{M}_{y}} & J_{z}^{-1}  \\
		\end{matrix}{{M}_{z}} \right]}^{T}}
\end{equation}
are the input matrix of the rotation. $\Theta \left(* \right)$ represents the diagonal matrix of vector $*$. The parameters are listed in Table \ref{table_1}.

System uncertainty is a primary factor affecting the FDD accuracy. The load uncertainty is considered as the major root cause of system uncertainty in this paper. Different from the existing work \cite{9684676} where the load uncertainty is treated a derivative-bounded variable, the relationship between load uncertainty and attitude of quadrotor UAV can be fully used.
\begin{equation}
	\label{eq_5_1}
{{\xi }_{\omega }}={{l}_{m}}\times mR_{EB}G=B_{\xi}^{*}{l}_m ,
\end{equation}
where ${l}_{m}$ represents the CoG offset vector.
\definecolor{mygray}{gray}{.9}
\begin{table}[!t]
	\renewcommand{\arraystretch}{2}
	\caption{Parameters of Quadrotor UAV Dynamics}
	\label{table_1}
	\centering
	\begin{tabular}{ll}
		\toprule
		Parameter & Meaning  \\
		\midrule
		$ {{p}_{E}}=[x \quad y \quad z]^T $ & The position of the quadrotor UAV in the IF\\
		\rowcolor{mygray}$ {{v}_{E}}=[u \quad v \quad w]^T $ & The linear velocity of the quadrotor UAV in the IF\\
		$ {\omega}=[p \quad q \quad r]^T $ & The angular rates of the quadrotor UAV in the BF\\
		\rowcolor{mygray}$ {{\xi }_{\omega}} $ & Unknown disturbance caused by load uncertainty  \\
		$ [{J}_{x} \quad {J}_{y} \quad {J}_{z}]^T $ & The moment of inertia in BF\\
		\rowcolor{mygray}$ [{M}_{x} \quad {M}_{y} \quad {M}_{z}]^T $ & The rolling torque, the pitching torque and the \\
		\rowcolor{mygray} & yawing torque\\
		$ {{F}_{m}} $ & The total thrust magnitude along $ z_{B} $\\
		\bottomrule
	\end{tabular} 
\end{table} 

\subsection{Actuator Model}
The body of a quadrotor UAV is driven by four brushless direct current motors (BLDCMs). Therefore, Given the physical configuration of the quadrotor UAV, the control inputs $u=\left[ \begin{matrix}
	{{F}_{m}} \quad {{M}_{x}} \quad {{M}_{y}} \quad {{M}_{z}}
\end{matrix} \right]^{T}$ can be simplified as follows
\begin{equation}
\label{eq_6}
\left[ \begin{matrix}
{{F}_{m}}  \\
{{M}_{x}}  \\
{{M}_{y}}  \\
{{M}_{z}}  \\
\end{matrix} \right]={{R}_{u}}{{f}_{s}}=\left[ \begin{matrix}
1 & 1 & 1 & 1  \\
-\frac{{{d}_{\phi }}}{2} & -\frac{{{d}_{\phi }}}{2} & \frac{{{d}_{\phi }}}{2} & \frac{{{d}_{\phi }}}{2}  \\
\frac{{{d}_{\theta }}}{2} & -\frac{{{d}_{\theta }}}{2} & \frac{{{d}_{\theta }}}{2} & -\frac{{{d}_{\theta }}}{2}  \\
{{c}_{\tau f}} & -{{c}_{\tau f}} & -{{c}_{\tau f}} & {{c}_{\tau f}}  \\
\end{matrix} \right]\left[ \begin{matrix}
{{f}_{1}}  \\
{{f}_{2}}  \\
{{f}_{3}}  \\
{{f}_{4}}  \\
\end{matrix} \right],
\end{equation}
where $f={{\left[ \begin{matrix}
			{{f}_{1}} & {{f}_{2}} & {{f}_{3}} & {{f}_{4}}  \\
		\end{matrix} \right]}^{T}}$ represents the forces generated by each rotor. The transformation matrix $R_u$ is used to convert $f$ into the control input vector $u$, which is used to control the quadrotor UAV. The physical constraints of each actuator are given by $0\le f_i\le {{f}{\max }}\left( i=1,\ldots ,4 \right)$, where ${f}{\max }$ is the maximum force generated by each rotor.

The actuator faults, which are modeled as loss of effectiveness (LoE) \cite{5966371}, can be represented as
\begin{equation}
	\label{eq_6_F1}
	f_d=\Gamma f,
\end{equation}
where $\Gamma$ is determined by a set of parameters $\Theta \left( {{\lambda }_{1}},{{\lambda }_{2}},{{\lambda }_{3}},{{\lambda }_{4}} \right)$, with $0 \le {{\lambda }_{i}} \le 1 \left( i=1,\ldots ,4 \right)$. The binary variable ${{\lambda }_{i}}$ represents the degree of fault in the $i\text{th}$ rotor. The condition ${{\lambda }_{i}}=0$ implies that the $i\text{th}$ rotor is fully damaged, while ${{\lambda }_{i}}=1$ denotes a healthy rotor. By incorporating this mapping matrix into the control system, the authors are able to effectively account for the impact of actuator degradation on the quadrotor UAV's performance

To simplify the analysis and control design, the model of BLDCM is approximated as a second-order system of the form ${{f}_{d}}={{\left( \frac{1}{Ts+1} \right)}^{2}}f$, where $f$ and $f_d$ are the input and output forces of the motor, respectively. This approximation allows the motor dynamics to be effectively represented by  the time constant $T$.

\subsection{Problem Statement}
The degradation of control inputs in a quadrotor UAV can be caused by two important factors: LoE faults and coil aging faults. To account for these factors, the actuator fault model can be summarized as follows
\begin{equation}
\label{eq_7_1}
f_{i}^{*}={{\lambda }_{i}}{{\left( \frac{1}{\left( {{T}_{a}}+{{T}_{c}} \right)s+1} \right)}^{2}}f_{i}^{c},
\end{equation}
where ${{f}_{i}}^{c}$ denotes the control command, and $f_{i}^{*}$ represents the actual output. 
\begin{figure*}[!t]
	\centering
	\includegraphics[width=3.5in]{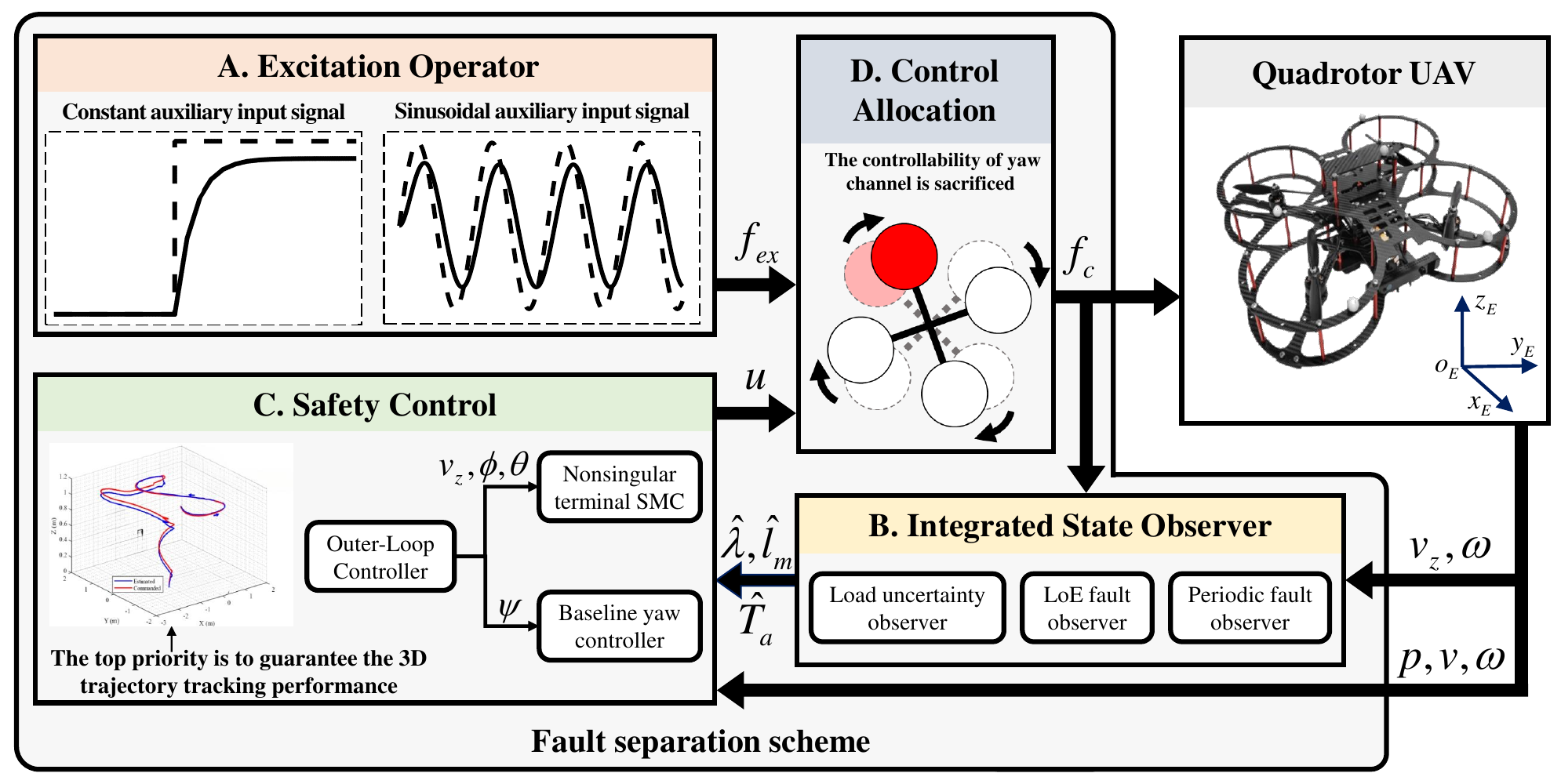}
	\caption{ System architecture of the proposed fault separation control for a Quadrotor UAV.}
	\label{fig_2}
\end{figure*}
 Consider the state vector $x_1={{\left[ \begin{matrix}
			{{v}_{z}} & p & q   \\
		\end{matrix} \right]}^{T}}$.  Based on the nonlinear model Eq. (\ref{eq_2}), Eq. (\ref{eq_6}), load uncertainty Eq. (\ref{eq_5_1}), and actuator fault model Eq. (\ref{eq_7_1}), it is rendered that
	\begin{equation}
		\label{eq_8}
		\left\{ \begin{array}{*{35}{l}}
			{{\left[ \begin{matrix}
						\dot{x_1} & \dot{r}   \\
					\end{matrix} \right]}^{T}}
				=f\left( x_1,r,t \right)\\
				\qquad \qquad \qquad+B{{R}_{u}} {{\Gamma }}{{\left[ \begin{matrix}
						{{f}_{1}}' & {{f}_{2}}' & {{f}_{3}}' & {{f}_{4}}'  \\
					\end{matrix} \right]}^{T}}+\xi  \\
			{{{{f_i}'}}}={{\left( \frac{1}{\left( {{T}_{a}}+{{T}_{c}} \right)s+1} \right)}^{2}}f_{i}^{c}, \left( i=1,\ldots ,4 \right)   \\
		\end{array} \right.,
	\end{equation}
	where $B=\Theta \left( \frac{\cos \phi \cos \theta }{m}, \frac{1}{{{J}_{x}}}, \frac{1}{{{J}_{y}}}, \frac{1}{{{J}_{z}}}\right)$. $\xi ={{\left[ \begin{matrix}
				0 & {{\xi }_{\omega }}  \\
			\end{matrix} \right]}^{T}}$ represents the load uncertainty. The known nonlinearity $f\left( x,t \right)$ is defined as
\begin{equation}
	\label{eq_9}
	f\left(  x_1,r,t \right)=\left[ \begin{matrix}
		-g  \\
		\frac{{{J}_{y}}-{{J}_{z}}}{{{J}_{x}}}qr  \\
		\frac{{{J}_{z}}-{{J}_{x}}}{{{J}_{y}}}pr  \\
		\frac{{{J}_{x}}-{{J}_{y}}}{{{J}_{z}}}pq  \\
	\end{matrix} \right].
\end{equation}
\newtheorem{assumption}{Assumption}
\begin{assumption}\label{assumption0}
	The LoE fault $\lambda_{i}$ and load uncertainty $l_m$ are considered as slowly changing values after occurrence, that is, $\dot{\lambda_{i}}\left(t>t_\lambda \right)=0$ and $\dot{l_m}\left(t>t_{l_m} \right)=0$
\end{assumption}

Actually, the physical structure of a quadrotor UAV is solid during mission. Thus, the root terms of LoE fault $\lambda_{i}$ and load uncertainty $l_m$ can be considered as constant values due to the relationship among actuator fault, load uncertainty, and system state is clarified. Moreover, based on Eqs. (\ref{eq_5_1}) and (\ref{eq_6_F1}), the upper bounds of $\lambda_{i}$ and $l_m$ can be given as
\begin{equation}
	\label{eq_12+1}
\left\{ \begin{aligned}
	&\lambda_{i} \le 1,\\
&	l_m \le \frac{1}{2}[L,W,H]^{T},
\end{aligned} \right.
\end{equation}
where $ [L,W,H]$ represents the maximum length, width, and height of a quadrotor UAV.

\begin{rem} \label{remark0}
As compared to the existing FDD research of quadrotot UAV focusing on the composite estimation of actuator faults, a deep coupling model of LoE faults, aging, and system uncertainty is established in this study. This step can pave a solid foundation for auxiliary input signals and FDD design, since the actuator faults information can be fully utilized.
\end{rem}

\section{FAULT SEPARATION SCHEME}
In this section, the design procedure of control system is presented, while the architecture is illustrated in Fig. \ref{fig_2}. The proposed fault separation scheme consists of an excitation operator, integrated state observer, a safety controller, and control allocation module. 



\subsection{Auxiliary Input Signal Design}
Injecting an appropriate auxiliary input signal into the system can amplify the impact of single fault or reduce the correlation between faults and uncertainties. Therefore, it  is an effective means to improve the FDD accuracy. Physical limits of actuator are common reasons for specifying constraints. Other reasons include ensuring that the auxiliary input signal does not affect flight safety while monitoring the fault status. 

Focusing on auxiliary input signal design, two factors need to be explicitly considered: 1) the physical limit of the actuator; and 2) the tracking performance for the desired reference. Practically, the auxiliary input signal is injected into each rotor. 

\textbf{Factor 1) } The physical limit of actuator.

 The amplitude of auxiliary input signal needs to meet the physical limit of the actuator. In addition, the balance of the quadrotor UAV in $z_{I}$-axis direction should be guaranteed. Therefore, the constraint on auxiliary input signal can be represented as
	\begin{equation}
	\label{eq_11}
	\max \left( {{\Delta }_{z}}-3{{f}_{\max }},0 \right)\le {{f}_{ex}}\le {{f}_{\max }},
\end{equation}
where the term ${{\Delta }_{z}}=\frac{mg}{\cos \phi \cos \theta }+m{z}^{(2)}_d$, ${z}^{(2)}_d$ refers to the maximum 2-nd order derivative of the desired reference in $z_{I}$-axis direction. $f_{ex}$ denotes any of the auxiliary input signals, which can be collectively referred to as the excitation operator.

\textbf{Factor 2) } The tracking performance.

Based on the dynamics of quadrotor UAV, the total thrust $F_m$ is determined by the desired acceleration, which is considered in Factor 1. Meanwhile, the control torques $M_x$, $M_y$, and $M_z$ are mainly influenced by the desired snap, which represents the 4-th order derivative of the trajectory. 

In order to obtain the maneuvering condition between the snap and the control torques, Eq. (\ref{eq_2}) is linearized at the equilibrium point. Thus, on the basis that the whole (or a period of time) desired reference is known,  the relationship between the snap and the control torque can be calculated as
	\begin{equation}
	\label{eq_14}
\left[ \begin{array}{*{35}{l}}
	x_{d}^{(4)}  \\
	y_{d}^{(4)}  \\
\end{array} \right]=R_t\left[ \begin{array}{*{35}{l}}
	{{M}_{x}}  \\
	{{M}_{y}}  \\
\end{array} \right],
\end{equation}
where $x_{d}^{(4)} $ and $y_{d}^{(4)}$ refer to the maximum 4-th order derivatives of the desired reference in $x_{I}$-axis and $y_{I}$-axis directions. 
\begin{equation*}
	R_t=\left[ \begin{matrix}
		-g{{S}_{\psi }}/{{J}_{x}} & -g{{C}_{\psi }}/{{J}_{y}}  \\
		g{{C}_{\psi }}/{{J}_{x}} & -g{{S}_{\psi }}/{{J}_{y}}  \\
	\end{matrix} \right].
\end{equation*}

Suppose M1 as the target rotor injected with auxiliary input signal. By ignoring the yaw channel, Eq. (\ref{eq_8}) under the condition of auxiliary input signal injection is rewritten as
	\begin{equation}
	\label{eq_12}
\left\{ \begin{aligned}
	& \dot{x}_1=\tilde{f}\left( x_1 \right)+\tilde{B}{{{\tilde{R}}}_{u}}{{\left[ \begin{matrix}
				{{\lambda }_{1}}{{{{f}'}}_{ex}} &{{{{f}'}}_{2}} & {{{{f}'}}_{3}} & {{{{f}'}}_{4}}  \\
			\end{matrix} \right]}^{T}}+\tilde{\xi } \\ 
	& {{f}_{i}}^{\prime }={{\left( \frac{1}{T_{a}s+1} \right)}^{2}}f_{i}^{c},\left( i=2,3,4 \right) \\ 
	& {{f}_{ex}}^{\prime }={{\left( \frac{1}{\left( {{T}_{a}}+{{T}_{c}} \right)s+1} \right)}^{2}}f_{ex}^{c} \\ 
\end{aligned} \right.,
\end{equation}
where $	\tilde{f}\left( x_1 \right)=\left[ \begin{matrix}
	-g ,	\frac{{{J}_{y}}-{{J}_{z}}}{{{J}_{x}}}qr  ,	\frac{{{J}_{z}}-{{J}_{x}}}{{{J}_{y}}}pr  
\end{matrix} \right]^{T}$, $\tilde{\xi }=[0,\xi _{\omega,1},\xi _{\omega,2}]^{T}$, $\tilde{B}=\Theta \left( \frac{\cos \phi \cos \theta }{m}, \frac{1}{{{J}_{x}}}, \frac{1}{{{J}_{y}}} \right)$, and 
\begin{equation*}
	\tilde{R}_u=\left[ \begin{matrix}
		1 & 1 & 1 & 1  \\
		-\frac{{{d}_{\phi }}}{2} & -\frac{{{d}_{\phi }}}{2} & \frac{{{d}_{\phi }}}{2} & \frac{{{d}_{\phi }}}{2}  \\
		\frac{{{d}_{\theta }}}{2} & -\frac{{{d}_{\theta }}}{2} & \frac{{{d}_{\theta }}}{2} & -\frac{{{d}_{\theta }}}{2} 
	\end{matrix} \right].
\end{equation*}

By combining Eq. (\ref{eq_14}) and Eq. (\ref{eq_12}), one can obtain
	\begin{equation}
	\label{eq_16}
\begin{aligned}
	 \left[ \begin{array}{*{35}{l}}
		x_{d}^{(4)}  \\
		y_{d}^{(4)}  \\
	\end{array} \right]=&{{R}_{t}}\left[ \begin{matrix}
		-\frac{{{d}_{\phi }}}{2} & \frac{{{d}_{\phi }}}{2} & \frac{{{d}_{\phi }}}{2}  \\
		-\frac{{{d}_{\theta }}}{2} & \frac{{{d}_{\theta }}}{2} & -\frac{{{d}_{\theta }}}{2}  \\
	\end{matrix} \right]\left[ \begin{array}{*{35}{l}}
		{{f}_{2}}  \\
		{{f}_{3}}  \\
		{{f}_{4}}  \\
	\end{array} \right] \\ 
	& +{{R}_{t}}\left(\left[ \begin{array}{*{35}{l}}
		-\frac{{{d}_{\phi }}}{2}{{f}_{ex}}  \\
		\frac{{{d}_{\theta }}}{2}{{f}_{ex}}  \\
	\end{array} \right]+\left[ \begin{array}{*{35}{l}}
		{{\xi }_{\omega ,1}}  \\
		{{\xi }_{\omega ,2}}  \\
	\end{array} \right] \right). 
\end{aligned}
\end{equation}

Based on Eq. (\ref{eq_16}) and the physical limit $0\le f_i\le {{f}_{\max }}\left( i=2,\ldots ,4 \right)$, the constraint on $f_{ex}$ can be calculated as 
\begin{equation}
		\label{eq_17}
	\left\{ \begin{aligned}
	-{{f}_{\max }}-{{\Delta }_{x}}& \le {{f}_{ex}}\le 2{{f}_{\max }}-{{\Delta }_{x}} \\ 
	-2{{f}_{\max }}-{{\Delta }_{y}}&\le {{f}_{ex}}\le {{f}_{\max }}-{{\Delta }_{y}}  
\end{aligned} \right.,
\end{equation}
where ${{\Delta }_{x}}=\frac{2}{{{d}_{\phi }}}\left( R_{t}^{-1}{{x_d}^{\left( 4 \right)}}-{{\xi }_{\omega ,1}} \right)$ and ${{\Delta }_{y}}=\frac{2}{{{d}_{\theta }}}\left( R_{t}^{-1}{{y_d}^{\left( 4 \right)}}-{{\xi }_{\omega ,2}} \right)$.

In summary, by combining Eq. (\ref{eq_11}) and Eq. (\ref{eq_17}), the auxiliary input signal must satisfy 
\begin{equation}
	\label{eq_18}
\left\{ \begin{aligned}
	& {{{\underline{f}}}_{ex}}\le {{f}_{ex}}\le {{{\overline{f}}}_{ex}} \\ 
	& {{{\underline{f}}}_{ex}}\text{=}\max \left( -{{f}_{\max }}-{{\Delta }_{x}},-2{{f}_{\max }}-{{\Delta }_{y}},{{\Delta }_{z}}-3{{f}_{\max }},0 \right) \\ 
	& {{{\overline{f}}}_{ex}}=\min \left( 2{{f}_{\max }}-{{\Delta }_{x}},{{f}_{\max }}-{{\Delta }_{y}},{{f}_{\max }} \right) \\ 
\end{aligned} \right..
\end{equation}

On the basis of the constraint on $f_{ex}$ in Eq. (\ref{eq_18}), two kinds of auxiliary input signals, including constant and sinusoidal, are generated by excitation operator. 

\textbf{Constant auxiliary input signal $f^{cons}_{ex}$}. Steady-state response of actuator for constant input can effectively avoid the influence of time constant. On the other hand, although the controllability of yaw channel is sacrificed, small yaw rate is helpful to maintain system safety and tracking performance. Therefore, in order to separate the LoE fault from the actuator aging fault, the constant auxiliary input signal is designed as
\begin{equation}
	\label{eq_19}
    f^{cons}_{ex}={\left( mg+z_{d}^{(2)} \right)}/{4}\;.
\end{equation}

According to the constant steady-state response result, the LoE fault $\lambda_{1}$ can be calculated as
\begin{equation}
	{{\lambda }_{1}}=\frac{z_{1}^{c}}{f_{ex}^{cons}},
\end{equation}
where $z_{1}^{c}$ represents the estimation result of LoE fault observer. It is worth noting that $f^{cons}_{ex}$ needs to meet the constraint of (\ref{eq_18}). 
\begin{figure}[!t]
	\centering
	\includegraphics[width=2.5in]{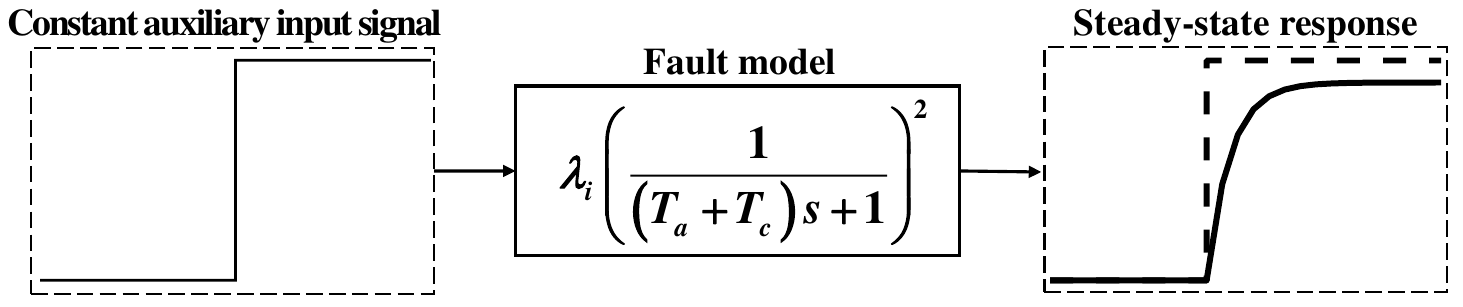}
	\caption{The steady-state response after the constant auxiliary input signal is injected.}
	\label{fig_con_response}
\end{figure}

\textbf{Sinusoidal auxiliary input signal $f^{sin}_{ex}$}. Similar to the constant auxiliary input signal, the steady-state response of a sinusoidal signal can be used to map the influence of unknown time constants onto the amplitude of a periodic signal. This property can be leveraged to achieve the separation and accurate estimation of unknown time constants by injecting sinusoidal auxiliary input signals into the system. Specifically, the sinusoidal auxiliary input signal is designed as follows
\begin{equation}
	\label{eq_21}
	f_{ex}^{\sin }=\frac{{{{\underline{f}}_{ex}}}+{{{\overline{f}}}_{ex}}}{2}\left( \alpha \sin \left( \beta t \right)+1 \right),
\end{equation}
where $\alpha \in (0,1] $ and $\beta$ are the parameters to be designed. Based on the steady-state result, The time constant that increases due to coil aging can be obtained
\begin{equation}
	{{T}_{c}}=\sqrt{\frac{\alpha \lambda_{1} \left( {{{{\underline{f}}_{ex}}}+{{{\overline{f}}}_{ex}}} \right)}{2{{\beta }^{2}}z_{1}^{p}}}-{{T}_{a}},
\end{equation}
where $z_{1}^{p}$ can be estimated by observer designed in \ref{sec1}.
\begin{figure}[!t]
	\centering
	\includegraphics[width=2.5in]{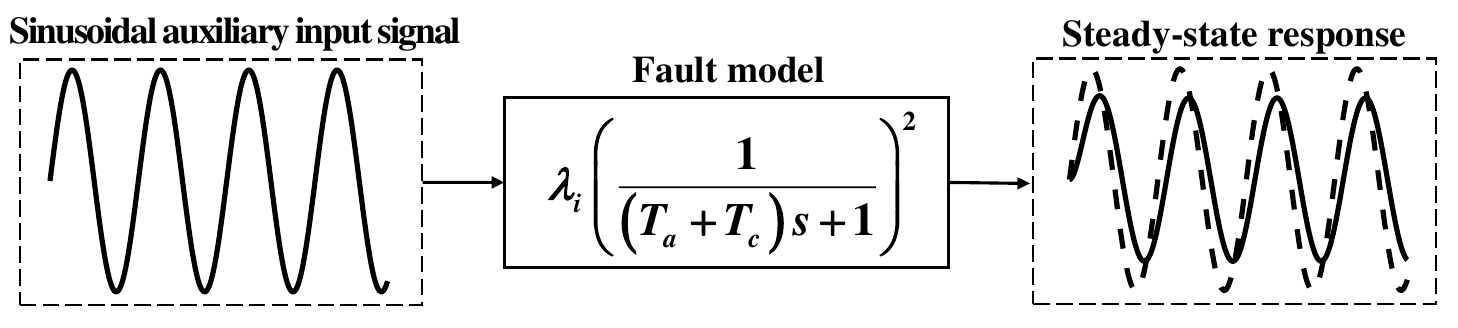}
	\caption{The steady-state response after the sinusoidal auxiliary input signal is injected.}
	\label{fig_sin_response}
\end{figure}
\begin{rem} \label{remark1}
	It is a common approach to calibrate a system according to its steady-state response. The innovations of this paper mainly lie in: Constrain calculation of auxiliary input signal by combining actuator physical limits and tracking performance for desired reference, and how to inject auxiliary input signal on an under-actuated quadrotor UAV system. In fact, the faults separation estimation can be achieved by any constant and sinusoidal auxiliary input signal meeting the constraints. This paper offers a feasible scheme of excitation operator design.
\end{rem}

\begin{rem} \label{remark2}
	As can be seen in Eq. (\ref{eq_8}), the LoE fault, aging, and load uncertainty are working together on the same channel. In particular, there is a coupling relationship between LoE fault and aging according to Eq. (\ref{eq_7_1}). Thus, the separability of LoE fault, aging, and load uncertainty cannot be ensured under the nominal system state and control input. In this paper, the constant auxiliary input signal is adopted to suppress the aging effect. The separability of LoE faults and load uncertainty is easy to be assured since the coupling is avoided. The LoE faults and load uncertainty can be estimated separately as a consequence. Subsequently, the sinusoidal auxiliary input signal is injected. By observing the sinusoidal steady-state response of each single actuator, the estimation of time constant resulting from aging can be achieved. 
\end{rem}


\subsection{Integrated State Observer Design}\label{sec1}
The proposed integrated state observer includes load uncertainty observer, LoE actuator fault observer, and periodic fault observer. With the sequential injection of  the auxiliary input signal, the load uncertainty, LoE actuator fault, and actuator aging can be estimated separately.

\textbf{Step 1)  The injection of constant auxiliary input signal.}

With the aid of constant auxiliary input signal, the influence of time constant caused by actuator aging is avoided. In that case, by recalling Eqs. (\ref{eq_5_1}) and (\ref{eq_8}), one can obtain
	\begin{equation}
	\label{eq_23}
{{\left[ \begin{matrix}
			{{{\dot{v}}}_{z}} & {{{\dot{x}}}_{2}}  \\
		\end{matrix} \right]}^{T}}=f\left( {{v}_{z}},{{x}_{2}} \right)+{B}{R}_{u}f+{{B}^{*}}{{d}_{\Gamma }}+B_{\xi }^{*}{{l}_{m}}, ,
\end{equation}
where the input matrix ${{B}^{*}}=B{{R}_{u}}\Theta ({{f}_{1}},{{f}_{2}},{{f}_{3}},{{f}_{4}})$, and ${{d}_{\Gamma }}={{\left[ \begin{matrix}
			{{\lambda }_{1}}-1 & {{\lambda }_{2}}-1 & {{\lambda }_{3}}-1 & {{\lambda }_{4}}-1  \\
		\end{matrix} \right]}^{T}}$ denotes the LoE fault vector.

The corresponding observers of the load uncertainty observer and LoE fault observer are designed as
\begin{equation}
	\label{eq_24}
\left\{ \begin{matrix}\begin{aligned}
	{{{\dot{z}}}_{1}}=&-{{k}_{1}}{{B}^{*}}{{z}_{1}}-{{k}_{1}}{{B}^{*}}{{p}_{1}}\left( {{v}_{z}},{{x}_{2}} \right)\\
	&-{{k}_{1}}\left[ B_{\xi }^{*}{{z}_{2}}+B_{\xi }^{*}{{p}_{2}}\left( {{v}_{z}},{{x}_{2}} \right)+f\left( {{v}_{z}},{{x}_{2}} \right)+{B}{R}_{u}f \right]  \\
	{{{\hat{d}}}_{\Gamma }}=&{{z}_{1}}+{{p}_{1}}\left( {{v}_{z}},{{x}_{2}} \right)  \\
	{{{\dot{z}}}_{2}}=&-{{k}_{2}}\tilde{B}_{\xi }^{*}{{z}_{2}}-{{k}_{2}}\tilde{B}_{\xi }^{*}{{p}_{2}}\left( {{x}_{2}} \right)\\
	&-{{k}_{2}}\left[ {\tilde{B}^{*}}{{z}_{1}}+{\tilde{B}^{*}}{{p}_{1}}\left( {{x}_{2}} \right)+f\left( {{x}_{2}} \right)+{B}{R}_{u}f \right]  \\
	{{{\hat{l}}}_{m}}=&{{z}_{2}}+{{p}_{2}}\left( {{x}_{2}} \right)  \\
\end{aligned}\end{matrix} \right.
\end{equation}
where $z_1$ and $z_2$ are the internal states of the load uncertainty observer and LoE fault observer. ${{{\hat{d}}}_{\Gamma }}$ and ${{{\hat{l}}}_{m}}$ represent the estimation results of ${{{{d}}}_{\Gamma }}$ and ${{{{l}}}_{m}}$. $\tilde{B}_{\xi }^{*} \in \mathbb{R}^{3 \times 4}$ and $\tilde{B}^{*}\in \mathbb{R}^{3 \times 4}$ indicate the corresponding sub-matrices of ${B}_{\xi }^{*}$ and  ${B}^{*}$. ${{p}_{1}}\left( {{v}_{z}},{{x}_{2}} \right)$ and ${{p}_{2}}\left( {{x}_{2}} \right)$ refer to the nonlinear function to be designed. The observer gains $k_1$ and $k_2$ are determined by
\begin{equation*}
	{{k}_{1}}=\left[ \frac{\partial {{p}_{1}}\left( {{v}_{z}},{{x}_{2}} \right)}{\partial {{v}_{z}}},\frac{\partial {{p}_{1}}\left( {{v}_{z}},{{x}_{2}} \right)}{\partial {{x}_{2}}} \right],\quad {{k}_{2}}=\left[ \frac{\partial {{p}_{2}}\left( {{x}_{2}} \right)}{\partial {{x}_{2}}} \right].
\end{equation*}
\newtheorem{theorem}{Theorem}
\begin{theorem}
	Consider the system described by Eq. (\ref{eq_23}) with the integrated observer Eq. (\ref{eq_24}). The terms of LoE fault $d_{\Gamma}$ and load uncertainty $l_m$ can be estimated separately if there exist observer gains $k_1$ and $k_2$ such that the maximum eigenvalues of the matrices $\left( -{{k}_{1}}{{B}^{*}}+k_{1}^{T}{{k}_{1}}+{{\left( {{{\tilde{B}}}^{*}} \right)}^{T}}{{{\tilde{B}}}^{*}} \right)$ and $\left( -{{k}_{2}}\tilde{B}_{\xi }^{*}+k_{2}^{T}{{k}_{2}}+{{\left( B_{\xi }^{*} \right)}^{T}}B_{\xi }^{*} \right)$ have negative real parts.
\end{theorem}
\begin{proof}
	A Lyapunov candidate is chosen as
\begin{equation}
	\label{eq_25}
	V_0=0.5\tilde{d}_{\Gamma }^{T}{{\tilde{d}}_{\Gamma }}+0.5\tilde{l}_{m}^{T}{{\tilde{l}}_{m}},
\end{equation}
where the observer estimation errors ${{\tilde{d}}_{\Gamma }}={{\hat{d}}_{\Gamma }}-d_{\Gamma}$, and ${{\tilde{l}}_{m }}={{\hat{l}}_{m }}-l_{m}$. By combining Eq. (\ref{eq_24}), the time-derivative of estiamtion error ${{\tilde{d}}_{\Gamma }}$ can be given as
\begin{equation}
	\label{eq_26}
\begin{aligned}
	{{{\dot{\tilde{d}}}}_{\Gamma }}=& {{{\dot{\hat{d}}}}_{\Gamma }}-{{{\dot{d}}}_{\Gamma }}={{{\dot{z}}}_{1}}+\left[ \frac{\partial {{p}_{1}}\left( {{v}_{z}},{{x}_{2}} \right)}{\partial {{v}_{z}}},\frac{\partial {{p}_{1}}\left( {{v}_{z}},{{x}_{2}} \right)}{\partial {{x}_{2}}} \right]\left[ \begin{matrix}
		{{{\dot{v}}}_{z}}  \\
		{{{\dot{x}}}_{2}}  \\
	\end{matrix} \right] \\ 
	 =&-{{k}_{1}}{{B}^{*}}{{z}_{1}}-{{k}_{1}}{{B}^{*}}{{p}_{1}}\left( {{v}_{z}},{{x}_{2}} \right)\\
	 &-{{k}_{1}}\left[ B_{\xi }^{*}{{z}_{2}}+B_{\xi }^{*}{{p}_{2}}\left( {{v}_{z}},{{x}_{2}} \right)+f\left( {{v}_{z}},{{x}_{2}} \right) \right]+{{k}_{1}}{{\left[ \begin{matrix}
				{{{\dot{v}}}_{z}} & {{{\dot{x}}}_{2}}  \\
			\end{matrix} \right]}^{T}} \\ 
	 =&-{{k}_{1}}{{B}^{*}}{{{\tilde{d}}}_{\Gamma }}-{{k}_{1}}B_{\xi }^{*}{{{\tilde{l}}}_{m}}. \\ 
\end{aligned}
\end{equation}

In an analogous manner, the time-derivative of estimation error ${{\tilde{l}}_{m }}$ can be obtained
\begin{equation}
		\label{eq_27}
	\begin{aligned}
		 {{{\dot{\tilde{l}}}}_{m}}&={{{\dot{\hat{l}}}}_{m}}-{\hat{l}_{m}}={{{\dot{z}}}_{2}}+\frac{\partial {{p}_{2}}\left( {{x}_{2}} \right)}{\partial {{x}_{2}}}{{{\dot{x}}}_{2}} \\ 
		 =&-{{k}_{2}}\tilde{B}_{\xi }^{*}{{z}_{2}}-{{k}_{2}}\tilde{B}_{\xi }^{*}{{p}_{2}}\left( {{x}_{2}} \right)\\
		 &-{{k}_{2}}\left[ {{{\tilde{B}}}^{*}}{{z}_{1}}+{{{\tilde{B}}}^{*}}{{p}_{1}}\left( {{x}_{2}} \right)+f\left( {{x}_{2}} \right) \right]+{{k}_{2}}{{{\dot{x}}}_{2}} \\ 
		 =&-{{k}_{2}}\tilde{B}_{\xi }^{*}{{{\tilde{l}}}_{m}}-{{k}_{2}}{\tilde{B}^{*}}{{{\tilde{d}}}_{\Gamma }}, \\ 
	\end{aligned}
\end{equation}

Differentiating $V_0$ yields
\begin{equation}
	\label{eq_28}
	\begin{aligned}
		 {{{\dot{V}}}_{0}}=&\tilde{d}_{\Gamma }^{T}{{{\dot{\tilde{d}}}}_{\Gamma }}+\tilde{l}_{m}^{T}{{{\dot{\tilde{l}}}}_{m}}\\
		 =&\tilde{d}_{\Gamma }^{T}\left( -{{k}_{1}}{{B}^{*}}{{{\tilde{d}}}_{\Gamma }}-{{k}_{1}}B_{\xi }^{*}{{{\tilde{l}}}_{m}} \right)+\tilde{l}_{m}^{T}\left( -{{k}_{2}}\tilde{B}_{\xi }^{*}{{{\tilde{l}}}_{m}}-{{k}_{2}}{{{\tilde{B}}}^{*}}{{{\tilde{d}}}_{\Gamma }} \right) \\ 
		 =&-\tilde{d}_{\Gamma }^{T}{{k}_{1}}{{B}^{*}}{{{\tilde{d}}}_{\Gamma }}-\tilde{d}_{\Gamma }^{T}{{k}_{1}}B_{\xi }^{*}{{{\tilde{l}}}_{m}}-\tilde{l}_{m}^{T}{{k}_{2}}\tilde{B}_{\xi }^{*}{{{\tilde{l}}}_{m}}-\tilde{l}_{m}^{T}{{k}_{2}}{{{\tilde{B}}}^{*}}{{{\tilde{d}}}_{\Gamma }}. \\ 
		 	\end{aligned}
		\end{equation}
	Moreover,  the following inequalities hold
\begin{equation}
		\label{eq_29}
	\begin{aligned}
		 -\tilde{d}_{\Gamma }^{T}{{k}_{1}}B_{\xi }^{*}{{{\tilde{l}}}_{m}}\le& 0.5{{\left( \tilde{d}_{\Gamma }^{T}{{k}_{1}} \right)}^{T}}\left( \tilde{d}_{\Gamma }^{T}{{k}_{1}} \right)+0.5{{\left( B_{\xi }^{*}{{{\tilde{l}}}_{m}} \right)}^{T}}\left( B_{\xi }^{*}{{{\tilde{l}}}_{m}} \right) \\ 
		\le&  0.5\tilde{d}_{\Gamma }^{T}k_{1}^{T}{{k}_{1}}{{{\tilde{d}}}_{\Gamma }}+0.5\tilde{l}_{m}^{T}{{\left( B_{\xi }^{*} \right)}^{T}}B_{\xi }^{*}{{{\tilde{l}}}_{m}}, \\ 
	\end{aligned}
\end{equation}	
\begin{equation}
		\label{eq_30}
\begin{aligned}
	 -\tilde{l}_{m}^{T}{{k}_{2}}{{{\tilde{B}}}^{*}}{{{\tilde{d}}}_{\Gamma }}\le& 0.5{{\left( \tilde{l}_{m}^{T}{{k}_{2}} \right)}^{T}}\left( \tilde{l}_{m}^{T}{{k}_{2}} \right)+0.5{{\left( {{{\tilde{B}}}^{*}}{{{\tilde{d}}}_{\Gamma }} \right)}^{T}}\left( {{{\tilde{B}}}^{*}}{{{\tilde{d}}}_{\Gamma }} \right) \\ 
	 \le &0.5\tilde{l}_{m}^{T}k_{2}^{T}{{k}_{2}}{{{\tilde{l}}}_{m}}+0.5\tilde{d}_{\Gamma }^{T}{{\left( {{{\tilde{B}}}^{*}} \right)}^{T}}{{{\tilde{B}}}^{*}}{{{\tilde{d}}}_{\Gamma }}. \\ 
\end{aligned}
\end{equation}		
	
	Substituting the inequalities (\ref{eq_29}) and (\ref{eq_30}) into Eq. (\ref{eq_28}) gives
	
\begin{equation}
	\label{eq_31}
	\begin{aligned}		 
		{{{\dot{V}}}_{0}} \le& -\tilde{d}_{\Gamma }^{T}{{k}_{1}}{{B}^{*}}{{{\tilde{d}}}_{\Gamma }}-\tilde{l}_{m}^{T}{{k}_{2}}\tilde{B}_{\xi }^{*}{{{\tilde{l}}}_{m}}+0.5\tilde{d}_{\Gamma }^{T}k_{1}^{T}{{k}_{1}}{{{\tilde{d}}}_{\Gamma }}\\
		 &+0.5\tilde{l}_{m}^{T}{{\left( B_{\xi }^{*} \right)}^{T}}B_{\xi }^{*}{{{\tilde{l}}}_{m}}+0.5\tilde{l}_{m}^{T}k_{2}^{T}{{k}_{2}}{{{\tilde{l}}}_{m}}\\
		 &+0.5\tilde{d}_{\Gamma }^{T}{{\left( {{{\tilde{B}}}^{*}} \right)}^{T}}{{{\tilde{B}}}^{*}}{{{\tilde{d}}}_{\Gamma }} \\ 
		=& \tilde{d}_{\Gamma }^{T}\left( -{{k}_{1}}{{B}^{*}}+k_{1}^{T}{{k}_{1}}+{{\left( {{{\tilde{B}}}^{*}} \right)}^{T}}{{{\tilde{B}}}^{*}} \right){{{\tilde{d}}}_{\Gamma }}\\
		&+\tilde{l}_{m}^{T}\left( -{{k}_{2}}\tilde{B}_{\xi }^{*}+k_{2}^{T}{{k}_{2}}+{{\left( B_{\xi }^{*} \right)}^{T}}B_{\xi }^{*} \right){{{\tilde{l}}}_{m}} .\\ 
	\end{aligned}
\end{equation}
Apparently, $\dot{V_0} \le 0$ can be guaranteed if the maximum eigenvalues of the matrices $\left( -{{k}_{1}}{{B}^{*}}+k_{1}^{T}{{k}_{1}}+{{\left( {{{\tilde{B}}}^{*}} \right)}^{T}}{{{\tilde{B}}}^{*}} \right)$ and $\left( -{{k}_{2}}\tilde{B}_{\xi }^{*}+k_{2}^{T}{{k}_{2}}+{{\left( B_{\xi }^{*} \right)}^{T}}B_{\xi }^{*} \right)$ have negative real parts, which completes the proof.
\end{proof}

\textbf{Step 2)  The injection of sinusoidal auxiliary input signal.}

In that case, the LoE fault is estimated separately in Step 1, and is considered known in Step 2. Therefore, by recalling Eqs. (\ref{eq_8}) and (\ref{eq_21}), one can achieve
	\begin{equation}
	\label{eq_25}
\left\{ \begin{matrix}\begin{aligned}
	{{\left[ {{{\dot{v}}}_{z}},{{{\dot{x}}}_{2}} \right]}^{T}}=f\left( {{v}_{z}},{{x}_{2}} \right)+B{{R}_{u}}\Gamma {{d}_{a}}+B_{\xi }^{*}{{l}_{m}}  \\
	{{d}_{a}}=\frac{A}{2\left( {{T}^{2}}{{\beta }^{2}}+1 \right)}\sin \left( \beta t-2\arctan \left( T\beta  \right) \right)  \\
\end{aligned}\end{matrix} \right.,
\end{equation}
where $A=0.5\alpha \left({{{{\underline{f}}_{ex}}}+{{{\overline{f}}}_{ex}}}\right)$. In addition, aging faults in the motor and load uncertainty can be estimated by analyzing the periodic signals with known frequencies and constant signals, respectively, that are generated by the injection of sinusoidal auxiliary input signals. The fault observer is designed as

\begin{equation}
	\label{eq_26}
	\left\{ \begin{matrix}\begin{aligned}
			   {{{\dot{z}}}_{3}}=&\left[ \beta -{{k}_{3}}B{{R}_{u}}\Gamma A  \right]{{z}_{3}}+\beta {{p}_{3}}\left( {{v}_{z}},{{x}_{2}} \right)\\
			   &-{{k}_{3}}\left[ B{{R}_{u}}\Gamma A {{p}_{3}}\left( {{v}_{z}},{{x}_{2}} \right)+B_{\xi }^{*}{{p}_{2}}\left( {{v}_{z}},{{x}_{2}} \right)+f\left( {{v}_{z}},{{x}_{2}} \right) \right]  \\
			{{{\hat{d}}}_{a}}=&\alpha \left( {{z}_{3}}+{{p}_{3}}\left( {{v}_{z}},{{x}_{2}} \right) \right)  \\
	\end{aligned}\end{matrix} \right.
\end{equation}
where $z_3$ is the estimation state of the designed observer, ${{{\hat{d}}}_{a }}$ represents the estimation result of ${{{{d}}}_{a }}$. The observer gain $k_3$ is determined as
\begin{equation*}
	{{k}_{3}}=\left[ \frac{\partial {{p}_{3}}\left( {{v}_{z}},{{x}_{2}} \right)}{\partial {{v}_{z}}},\frac{\partial {{p}_{3}}\left( {{v}_{z}},{{x}_{2}} \right)}{\partial {{x}_{2}}} \right].
\end{equation*}
The separability of a periodic signal with known frequencies and a constant signal using DO has been discussed and proved in previous work \cite{guo2020multiple}. The details substantiating the separation estimation between signals from an aging system fault and signals from load uncertainty are excluded from this paper due to space limitations.

\subsection{Safety Controller Design}

The quadrotor UAV is a typical underactuated system that cannot achieve simultaneous control of the position state $p_E$ and attitude state ${{\Omega }{b}}$.  However, injecting an excitation operator to the target rotor can result in a loss of controllability. To overcome this issue, a safety controller is designed to ensure the stability of the system by sacrificing the controllability of the yaw channel. A nonsingular terminal sliding mode control (SMC) is then designed to guarantee the tracking performance of the pitch and roll channels. Finally, the excitation operator is injected into the control allocation module to guarantee 3D trajectory tracking performance.

\textbf{The nonsingular terminal SMC}: The controllability of the yaw channel is sacrified when the excitation operator is injected. Aiming at the roll and pitch channels, define the state vectors ${x_3}={{[ \begin{matrix}
			\varphi  & \theta   \\
		\end{matrix} ]^{T}}}$ and control input $u_1={{[ \begin{matrix}
		{{M }_{x}} & {{M}_{y}}  \\
	\end{matrix} ]^{T}}}$. The nonlinear quadrotor UAV model (\ref{eq_2}) can be expressed as
\begin{equation}
	\label{eq_37}
	\left\{ \begin{matrix}\begin{aligned}
			{{{\dot{x}}}_{3}}&={{x}_{4}}+{{d}_{\psi }}  \\
			{{{\dot{x}}}_{4}}&=a\left( x \right)+bu_1+{{d}_{0}}  \\
	\end{aligned}\end{matrix} \right.,
\end{equation}
where ${x_3}={{[ \begin{matrix}
			p & q  \\
		\end{matrix} ]^{T}}}$ denotes the angular rates along $x$ and $y$ axes. $b=\Theta \left( {1}/{{{I}_{x}}}\;,{1}/{{{I}_{y}}}\right)$ represents the input matrix. $d_\psi$ is the mismatched disturbance resulting from yaw rotation. $d_0={{B}^{*}}{{d}_{\Gamma }}+B_{\xi }^{*}{{l}_{m}}$ refers to the composite term of actuator faults, aging, and load uncertainty, which can be estimated through the proposed excitation operator. $a(x)={{\left[ 
		\left( {{I}_{y}}-{{I}_{z}} \right)qr/{{I}_{x}}\  \quad \left( {{I}_{z}}-{{I}_{x}} \right)pr/{{I}_{y}}  \right]}^{T}}$ refers to the  nonlinear function. 

\begin{assumption}\label{assumption1}
	The mismatched disturbance $d_\psi$ is differentiable, and ${d_\psi}$ has a Lipschitz constant $L$.
\end{assumption}

Based on Assumption \ref{assumption1}, a finite-time DO, proposed by \cite{li2014disturbance}, is used to estimate $d_\psi$ in this paper.
\begin{equation}
	\label{eq_38}
	\left\{ \begin{aligned}
		& {{{\dot{z}}}_{0}}={{v}_{0}}+{{x}_{4}},\quad {{{\dot{z}}}_{1}}={{v}_{1}} \\ 
		& {{v}_{0}}=-{{\lambda }_{0}}{{L}^{{1}/{2}\;}}{{\left| {{z}_{0}}-{{x}_{3}} \right|}^{{1}/{2}\;}}\text{sgn} \left( {{z}_{0}}-{{x}_{3}} \right)+{{z}_{1}} \\ 
		& {{v}_{1}}=-{{\lambda }_{1}}L\text{sgn} \left( {{z}_{1}}-{{v}_{0}} \right) \\ 
		& {{{\hat{x}}}_{3}}={{z}_{0}},\quad {{{\hat{d}}}_{\psi }}={{z}_{1}} \\ 
	\end{aligned} \right.,
\end{equation}
where ${{\lambda }_{i}}>0\left( i=1,2 \right) $ represents the observer gain. 

Define the tracking error of Eq. (\ref{eq_37}) $e_3=x_3-x_{3,d}$ and $e_4=x_4-x_{4,d}$. According to the estimation result of Eq. (\ref{eq_38}), the terminal sliding surface is designed as
\begin{equation}
	\label{eq_39}
	s={{e}_{3}}+{\gamma }{{\left( {{e}_{4}}+{{{\hat{d}}}_{\psi }} \right)}^{{\varepsilon_1}/{\varepsilon_2}\;}},
\end{equation}
where $\gamma$ is a positive constant diagonal matrix. $\varepsilon_1$ and $\varepsilon_2$ represent the positive constant to be designed satisfying $1<{\varepsilon_1}/{\varepsilon_2}<2$.

According to the designed sliding mode surface (\ref{eq_39}), by resorting to the nonsingular terminal SMC, the control laws can be designed as
\begin{equation}
	\label{eq_40}
	\begin{aligned}
		{{u}_{1}}=&-{{b}^{-1}}\left[ a\left( x \right)+{{\gamma }^{-1}}\frac{{{\varepsilon }_{2}}}{{{\varepsilon }_{1}}}{{\left( {{e}_{4}}+{{{\hat{d}}}_{\psi }} \right)}^{2-{{\varepsilon }_{1}}/{{\varepsilon }_{2}}}}+{{{\hat{d}}}_{0}}\right]\\
		&-{{b}^{-1}}\left[-{{v}_{1}}-k\text{sgn} \left( s \right)\right] \end{aligned} 
\end{equation}
where $\varepsilon_3 \in \left( 0 \quad 1 \right)$ is the designed positive constant.

\begin{theorem}
	Consider the dynamic model of the quadrotor UAV Eq. (\ref{eq_37}) and Assumption \ref{assumption1}. Based on the proposed control law (\ref{eq_21}), the tracking error $e_3$ and $e_4$ will converge into zero, and the closed-loop system Eq. (\ref{eq_37}) is input-to-state stable.
\end{theorem}

\begin{proof}
Taking the derivative of the proposed sliding mode surface Eq. (\ref{eq_39}) yields
\begin{equation}
	\label{eq_41}
	\begin{aligned}
		 \dot{s}=&{{{\dot{e}}}_{3}}+\gamma \frac{{{\varepsilon }_{1}}}{{{\varepsilon }_{2}}}\Theta\left({{\left( {{e}_{4}}+{{{\hat{d}}}_{\psi }} \right)}^{\frac{{{\varepsilon }_{1}}}{{{\varepsilon }_{2}}}-1}}\right)\left( {{{\dot{e}}}_{4}}+{{{\dot{\hat{d}}}}_{\psi }} \right) \\ 
		 =&{{e}_{4}}+{{d}_{\psi }}+\gamma \frac{{{\varepsilon }_{1}}}{{{\varepsilon }_{2}}}\Theta\left({{\left( {{e}_{4}}+{{{\hat{d}}}_{\psi }} \right)}^{\frac{{{\varepsilon }_{1}}}{{{\varepsilon }_{2}}}-1}}\right)\left( a\left( x \right)+b{{u}_{1}} \right) \\ 
		&+\gamma \frac{{{\varepsilon }_{1}}}{{{\varepsilon }_{2}}}\Theta\left({{\left( {{e}_{4}}+{{{\hat{d}}}_{\psi }} \right)}^{\frac{{{\varepsilon }_{1}}}{{{\varepsilon }_{2}}}-1}}\right)\left({{d}_{0}}+{{{\dot{\hat{d}}}}_{\psi }}\right).
	\end{aligned}
\end{equation}

By applying the proposed nonsingular terminal SMC law Eq. (\ref{eq_21}), Eq. (\ref{eq_41}) can be transformed into
\begin{equation}
	\label{eq_42}
\dot{s}=-\gamma \frac{{{\varepsilon }_{1}}}{{{\varepsilon }_{2}}}\Theta\left({{\left( {{e}_{4}}+{{{\hat{d}}}_{\psi }} \right)}^{\frac{{{\varepsilon }_{1}}}{{{\varepsilon }_{2}}}-1}}\right)\left( k\text{sgn} \left( s \right){{\left| s \right|}^{{{\varepsilon }_{3}}}}+{{e}_{{{d}_{0}}}} \right)-{{e}_{{{d}_{\psi }}}},
\end{equation}
where ${{e}_{{{d}_{0}}}}=\hat{d}_0-d_0={{B}^{*}}{{\tilde{d}}_{\Gamma }}+B_{\xi }^{*}{{\tilde{l}}_{m}}$ and ${{e}_{{{d}_{\psi}}}}=\hat{d}_\psi-d_\psi$ represents the estimation errors of $d_0$ and $d_\psi$. 

Considering the following Lyapunov candidate
\begin{equation}
	\label{eq_43}
{{V}_{1}}=\frac{1}{2}{{s}^{T}}s+\frac{1}{2}e_{3}^{T}{{e}_{3}}+\frac{1}{2}{\left( {{e}_{4}}+{{{\hat{d}}}_{\psi }} \right)}^{T}{\left( {{e}_{4}}+{{{\hat{d}}}_{\psi }} \right)}.
\end{equation}

The time-derivative of $V_1$ can be given as
\begin{equation}
	\label{eq_44}
	\begin{aligned}
		 {{{\dot{V}}}_{1}}=&{{s}^{T}}\dot{s}+e_{3}^{T}{{{\dot{e}}}_{3}}+{{\left( {{e}_{4}}+{{{\hat{d}}}_{\psi }} \right)}^{T}}\left( {{{\dot{e}}}_{4}}+{{{\dot{\hat{d}}}}_{\psi }} \right) \\ 
		 =&{{s}^{T}}\left[ -\gamma \frac{{{\varepsilon }_{1}}}{{{\varepsilon }_{2}}}\Theta\left({{\left( {{e}_{4}}+{{{\hat{d}}}_{\psi }} \right)}^{\frac{{{\varepsilon }_{1}}}{{{\varepsilon }_{2}}}-1}}\right)\left( k\text{sgn}\left( s \right)+{{e}_{{{d}_{0}}}} \right) \right]\\
		 &-{{s}^{T}}{{e}_{{{d}_{\psi }}}} +e_{3}^{T}\left( {{e}_{4}}+{{d}_{\psi }} \right) \\
		 &-{{\left( {{e}_{4}}+{{{\hat{d}}}_{\psi }} \right)}^{T}}\left({{\gamma }^{-1}}\frac{{{\varepsilon }_{2}}}{{{\varepsilon }_{1}}}{{\left( {{e}_{4}}+{{{\hat{d}}}_{\psi }} \right)}^{\frac{{{\varepsilon }_{1}}}{{{\varepsilon }_{2}}}-1}}\right)\\
		 & -{{\left( {{e}_{4}}+{{{\hat{d}}}_{\psi }} \right)}^{T}}\left( {{e}_{{{d}_{0}}}}+k\text{sgn}\left( s \right) \right) \\ 
		\le & {{s}^{T}}\left[ -\gamma \frac{{{\varepsilon }_{1}}}{{{\varepsilon }_{2}}}{\Theta\left({\left( {{e}_{4}}+{{{\hat{d}}}_{\psi }} \right)}^{\frac{{{\varepsilon }_{1}}}{{{\varepsilon }_{2}}}-1}\right)}{{e}_{{{d}_{0}}}}-{{e}_{{{d}_{\psi }}}} \right]\\
		&+e_{3}^{T}\left( {{e}_{4}}+{{d}_{\psi }} \right)+{{\left( {{e}_{4}}+{{{\hat{d}}}_{\psi }} \right)}^{T}}\left( -{{e}_{{{d}_{0}}}}-k\text{sgn}\left( s \right) \right). \\ 
	\end{aligned}
\end{equation}		
		
		Define  the maximum eigenvalues of a matrix ${{\lambda }_{\max }}\{\centerdot\}$, one can thereby achieve
		\begin{equation}
			\label{eq_44_1}
				\begin{aligned}
		{{{\dot{V}}}_{1}} \le& {{\lambda }_{\max }}\{\gamma\} \frac{{{\varepsilon }_{1}}}{{{\varepsilon }_{2}}}\left| s \right|\left( 1+\left| {{e}_{4}}+{{{\hat{d}}}_{\psi }} \right| \right)\left| {{e}_{{{d}_{0}}}} \right|+\left| {{s}^{T}}{{e}_{{{d}_{\psi }}}} \right|\\
		 &+\left| e_{3}^{T}\left( {{e}_{4}}+{{d}_{\psi }} \right) \right|+{{\left| {{e}_{4}}+{{{\hat{d}}}_{\psi }} \right|}^{T}}\left(\left| {{e}_{{{d}_{0}}}}\right|+kI\right)  \\ 
		 \le&  {{\lambda }_{\max }}\{\gamma\} \frac{{{\varepsilon }_{1}}}{2{{\varepsilon }_{2}}}\left( \left( 1+\left| {{e}_{{{d}_{0}}}} \right| \right){{s}^{T}}s+{{\left| {{e}_{{{d}_{0}}}} \right|}^{2}}\right)\\
		&+  {{\lambda }_{\max }}\{\gamma\} \frac{{{\varepsilon }_{1}}}{2{{\varepsilon }_{2}}}\left| {{e}_{{{d}_{0}}}} \right|{{\left( {{e}_{4}}+{{{\hat{d}}}_{\psi }} \right)}^{T}}\left( {{e}_{4}}+{{{\hat{d}}}_{\psi }} \right) +\frac{1}{2}{{s}^{T}}s\\
		&+\frac{1}{2}e_{{{d}_{\psi }}}^{T}{{e}_{{{d}_{\psi }}}} +\frac{1}{2}e_{3}^{T}{{e}_{3}}+\frac{1}{2}{{\left( {{e}_{4}}+{{{\hat{d}}}_{\psi }} \right)}^{T}}\left( {{e}_{4}}+{{{\hat{d}}}_{\psi }} \right)\\
		&+\frac{1}{2}{{\left( {{e}_{4}}+{{{\hat{d}}}_{\psi }} \right)}^{T}}\left( {{e}_{4}}+{{{\hat{d}}}_{\psi }} \right)+\frac{1}{2}{{\left| {{e}_{{{d}_{0}}}} \right|}^{2}}\\
		&+\frac{k}{2}{{\left( {{e}_{4}}+{{{\hat{d}}}_{\psi }} \right)}^{T}}\left( {{e}_{4}}+{{{\hat{d}}}_{\psi }} \right)+\frac{k}{2} \\ 
		 \le &{{k}_{n}}{{V}_{1}}+{{\eta }_{n}}, \\ 
	\end{aligned}
\end{equation}
where ${{k}_{n}}={{\lambda }_{\max }}\{\frac{{{\varepsilon }_{1}}}{{{\varepsilon }_{2}}}\left( 1+{{\left| {{e}_{{{d}_{0}}}} \right|}^{2}} \right)\gamma +I,\frac{{{\varepsilon }_{1}}}{{{\varepsilon }_{2}}}\left| {{e}_{{{d}_{0}}}} \right|\gamma +\left( k+2 \right)I,1\}$, and ${{\eta }_{n}}=\left( \frac{{{\lambda }_{\max }}\{\gamma \}{{\varepsilon }_{1}}}{2{{\varepsilon }_{2}}}+\frac{1}{2} \right){{\left| {{e}_{{{d}_{0}}}} \right|}^{2}}+\frac{1}{2}{{\left| {{e}_{{{d}_{\psi }}}} \right|}^{2}}+\frac{k}{2}$. Thus, based on the analysis in \cite{YANG20132287}, it can be seen from Eq. (\ref{eq_44}) that the sliding mode surface $s$, tracking error $e_3$, and ${{e}_{4}}+{{{\hat{d}}}_{\psi }}$ will not escape to infinity in finite-time since the boundedness of ${{e}_{{{d}_{0}}}}$ and ${{e}_{{{d}_{\psi}}}}$ .

According to Theorem 9.2 in \cite{li2014disturbance}, the mismatched disturbance ${{{{d}}}_{\psi }}$ can be estimated within finite-time through the finite-time disturbance observer. In that case, the derivative of the proposed sliding mode surface Eq. (\ref{eq_42}) can be rewritten as
\begin{equation}
	\label{eq_45}
	\dot{s}=-\gamma \frac{{{\varepsilon }_{1}}}{{{\varepsilon }_{2}}}\Theta\left(\tilde{e}_4\right)\left( k\text{sgn} \left( s \right)+{{e}_{{{d}_{0}}}} \right),
\end{equation}
where $\tilde{e}_4={{\left( {{e}_{4}}+{{{\hat{d}}}_{\psi }} \right)}^{\frac{{{\varepsilon }_{1}}}{{{\varepsilon }_{2}}}-1}}$. Consider another Lyapunov function
\begin{equation}
	\label{eq_46}
	{{V}_{2}}=\frac{1}{2}{{s}^{T}}s+V_0.
\end{equation}

Differentiating $V_2$ renders
\begin{equation}
		\label{eq_47}
\begin{aligned}
	{{{\dot{V}}}_{2}}=& {{s}^{T}}\dot{s}+{{{\dot{V}}}_{0}}\\
	=& -{{s}^{T}}\gamma \frac{{{\varepsilon }_{1}}}{{{\varepsilon }_{2}}}\Theta \left( {{{\tilde{e}}}_{4}} \right)\left( k\text{sgn}\left( s \right)+{{B}^{*}}{{{\tilde{d}}}_{\Gamma }}+B_{\xi }^{*}{{{\tilde{l}}}_{m}} \right)+{{\dot{V}}_{0}}\\
	\le & -\frac{{{\varepsilon }_{1}}}{{{\varepsilon }_{2}}}\Theta \left(\gamma  {{{\tilde{e}}}_{4}} \right)\left( k-{{B}^{*}}{{{\tilde{d}}}_{\Gamma }}-B_{\xi }^{*}{{{\tilde{l}}}_{m}} \right)|s|+{{\dot{V}}_{0}}.
\end{aligned}
\end{equation}

Considering Eqs. (\ref{eq_5_1}), (\ref{eq_6_F1}) and (\ref{eq_12+1}), the following inequality holds:
\begin{equation}
	\label{eq_49}
\left| {{B}^{*}}{{{\tilde{d}}}_{\Gamma }}+B_{\xi }^{*}{{{\tilde{l}}}_{m}} \right|\le -B{{R}_{u}}\left( {{f}_{\max }}I \right)+\left[ \begin{matrix}
	0  \\
	0.5B_{\xi }^{*}{{\left[ L,W,H \right]}^{T}} . \\
\end{matrix} \right]
\end{equation}
Hence, based on (\ref{eq_31}), it is known that $\dot{V}_2 \le 0$ due to $k \ge -B{{R}_{u}}\left( {{f}_{\max }}I \right)+\left[ \begin{matrix}
	0  \\
	0.5B_{\xi }^{*}{{\left[ L,W,H \right]}^{T}}  \\
\end{matrix} \right]$. As a result, the sliding surface $s=0$ can be reached,

Once the sliding surface $s=0$ is reached, it is derived from the sliding surface (\ref{eq_39}) and the system dynamics (\ref{eq_37}) that
\begin{equation}
	\label{eq_50}
s={{e}_{3}}+\gamma {{\left( {{e}_{4}}+{{{\hat{d}}}_{\psi }} \right)}^{{{\varepsilon }_{1}}/{{\varepsilon }_{2}}}}={{e}_{3}}+\gamma \dot{e}_{3}^{{{\varepsilon }_{1}}/{{\varepsilon }_{2}}}=0.
\end{equation}

With the chosen control parameters, the tracking errors $e_3$ and $e_4$ of system (\ref{eq_50}) will converge into zero, and the closed-loop system (\ref{eq_50}) is input-to-state stable, which completes the proof.
\end{proof}	

\subsection{Control Allocation Module}
	The control allocation unit is responsible for translating the control input vector $u={{\left[ \begin{matrix}
				{{F}{m}} & {{M}{x}} & {{M}{y}} & {{M}{z}} \
			\end{matrix} \right]}^{T}}$ into the corresponding thrust commands for each rotor, as defined by Eq. (\ref{eq_6}). The proposed fault separation maneuver involves sacrificing the controllability of the yaw channel to inject auxiliary input signals. The assumption is that the thrust command of the target rotor receiving the auxiliary input signal is already known. Under this condition, the mapping matrix can be redefined as follows.
\begin{equation}
	\label{eq_4_1}
	\left[ \begin{matrix}
		{{F}_{m}}  \\
		{{M}_{x}}  \\
		{{M}_{y}}  \\
	\end{matrix} \right]=\left[ \begin{matrix}
		1 & 1 & 1  \\
		-\frac{{{d}_{\phi }}}{2} & \frac{{{d}_{\phi }}}{2} & \frac{{{d}_{\phi }}}{2}  \\
		 -\frac{{{d}_{\theta }}}{2} & \frac{{{d}_{\theta }}}{2} & -\frac{{{d}_{\theta }}}{2} 
	\end{matrix} \right]\left[ \begin{matrix}
		{{f}_{2}}  \\
		{{f}_{3}}  \\
		{{f}_{4}}  \\
	\end{matrix} \right]+f_{ex}\left[ \begin{matrix}
     1    \\
	-\frac{{{d}_{\phi }}}{2}  \\
	\frac{{{d}_{\theta }}}{2} \\
\end{matrix} \right].
\end{equation}
Moreover, the thrust force sent to each actuator can obtain
\begin{equation}
	\label{eq_4_2}
	\left[ \begin{matrix}
		{{f}^{c}_{1}}  \\
		{{f}^{c}_{2}}  \\
		{{f}^{c}_{3}}  \\
		{{f}^{c}_{4}}  \\
	\end{matrix} \right]=\left[ \begin{matrix}
		{{f}_{ex}}  \\
		\left[ \begin{matrix}
			0.5 & -\frac{1}{{{d}_{\phi }}}  & 0  \\
		    0.5 & 0                                         & \frac{1}{{{d}_{\theta }}}  \\
			0    &   \frac{1}{{{d}_{\phi }}}  & -\frac{1}{{{d}_{\theta }}}  
		\end{matrix} \right] 	\left[ \begin{matrix}
		{{F}_{m}}-	 {{f}_{ex}}  \\
		{{M}_{x}}+\frac{{{d}_{\phi }}}{2} 	{{f}_{ex}}   \\
		{{M}_{y}} -\frac{{{d}_{\theta }}}{2}  {{f}_{ex}} \\
	\end{matrix} \right]\\
	\end{matrix} \right].
\end{equation}

The bijective relationship between the control input generated by the fault separation controller and the thrust command of each rotor ensures that the injection of auxiliary input signals does not compromise the safety and tracking performance of the quadrotor UAV in the lift, pitch, and roll channels. This is because the thrust command for each rotor is directly controlled by the control input, which is designed to satisfy the desired lift, pitch, and roll commands. Therefore, any changes in the thrust command of a single rotor due to the injection of auxiliary input signals will not affect the overall safety and stability of the quadrotor UAV.


\begin{rem} \label{remark3}
	On one hand, by combining the excitation operator and integrated state observer, the coupled LoE fault, aging and load uncertainty can be estimated separately. Compared with passive FDD, the robustness and accuracy of FDD can be improved. On the other hand, the injection of  excitation operator leads to the unstability of the under-actuated quadrotor UAV system, which is a common problem caused by  auxiliary input signals. In an attempt to address this problem, the controllability of yaw channel is sacrificed. In view of (\ref{eq_40}), the term of sliding surface significantly compensates the mismatched disturbance caused by yaw rotation. Moreover, the tracking of pitch and roll can be guaranteed. As a result, the 3-D trajectory tracking performance during the injection of excitation operator can be ensured.
\end{rem}

\section{EXPERIMENTAL VALIDATION}
In this section, comparative simulations and real-world tests are conducted to verify the effectiveness of the proposed strategy. 

\subsection{Simulation Results  }

This section shows some simulation results to verify whether the proposed scheme can effectively address the actuator faults and unknown input delay separately. Futhermore, the proposed fault separation and estimation scheme, the compared fixed-time FDD (named FxT) \cite{9385897}, and baseline robust controller without FDD (named NoFDD) are compared under the same condition to illustrate the superiority of the proposed scheme.

\subsection{Simulation Condition}
The simulation are designed as follows

$\bullet$ The physical parameters of the quadrotor UAV considered in the simulation are listed in Table \ref{table_2}. 

$\bullet$ The test trajectory is designed as $[0.7*\sin(\frac{2\pi}{T}) \quad 0.7*\cos(\frac{2*\pi}{T}) \quad 0.1*(t-20)+1]$, where the period $T=4\pi$. The tracking start time is set as the 20\emph{th} second.

$\bullet$ The simulation lasts 50s. The actuator faults occurs in the 35\emph{th} second of the simulation, and the LoE magnitude is 30\%. Meanwhile, the unknown input delay is considered to always exist and constant throughout the simulation.

\definecolor{mygray}{gray}{.9}
\begin{table}[!t]
	\renewcommand{\arraystretch}{1.3}
	\caption{Physical Parameters of the Quadrotor UAV } 
	\label{table_2}
	\centering
	\begin{tabular}{lll}
		\toprule
		Parameter & Value & Unit \\
		\midrule
		$ m $ & 1.121 & $kg$\\
		\rowcolor{mygray} $ g $ & 9.81 & $m/s^2$\\
		$ {d}_{\phi } $ & 0.2122 & $m$\\
		\rowcolor{mygray} ${{d}_{\theta }}$ & 0.2122  & $m$\\
		$ {{I}_{x}}$ & 5.6$\times$10$^{-3}$ & $kgm^2$\\
		\rowcolor{mygray} $ {{I}_{y}}$ & 5.6$\times$10$^{-3}$ & $kgm^2$\\
		$ {{I}_{z}}$ & 8.1$\times$10$^{-3}$ & kgm$^2$ \\
		\bottomrule
	\end{tabular}
\end{table}

\begin{figure}[!t]
	\centering
	\includegraphics[width=3.5in]{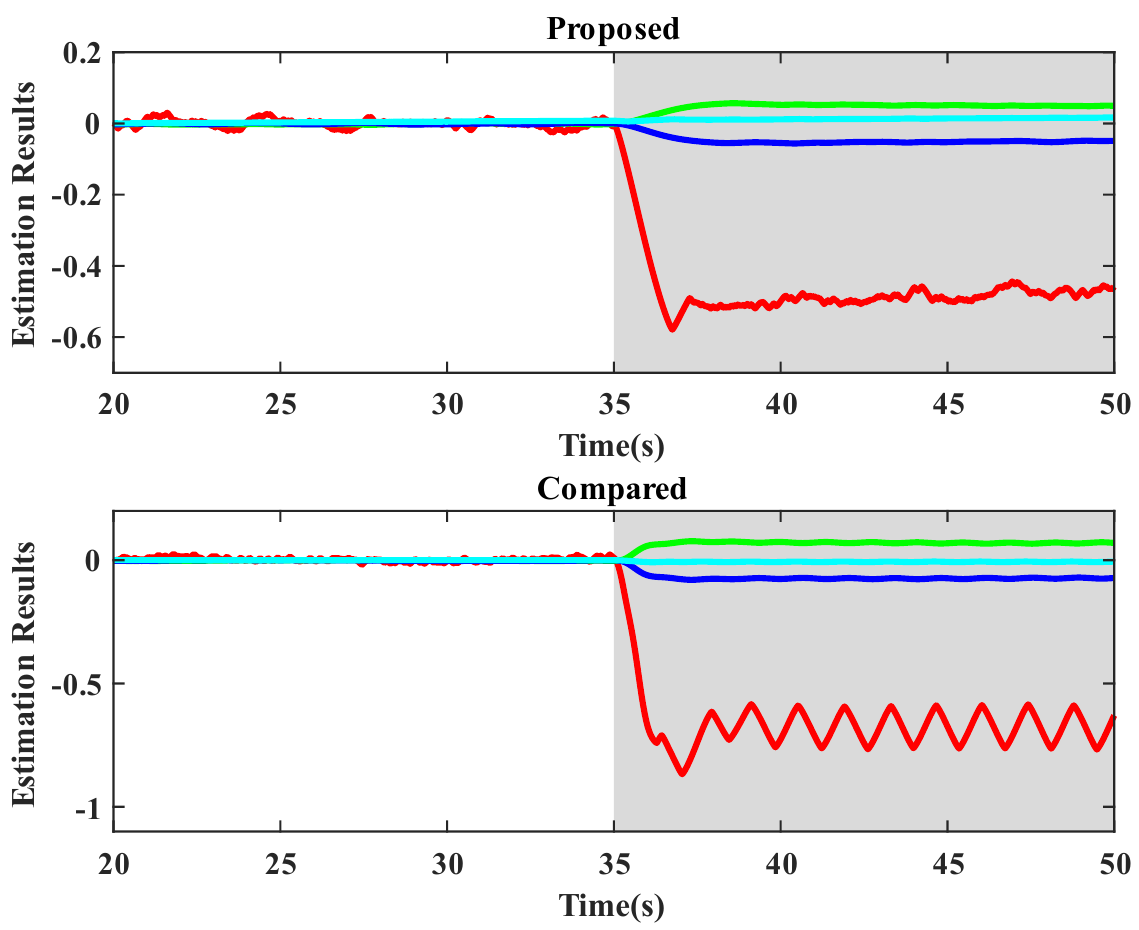}
	\caption{Comparison of estimation results: the gray area denotes the scenarios of LoE actuator fault occurrence}
	\label{fig_3}
\end{figure}

\begin{figure*}[!t]
	\centering
	\includegraphics[width=6.5in]{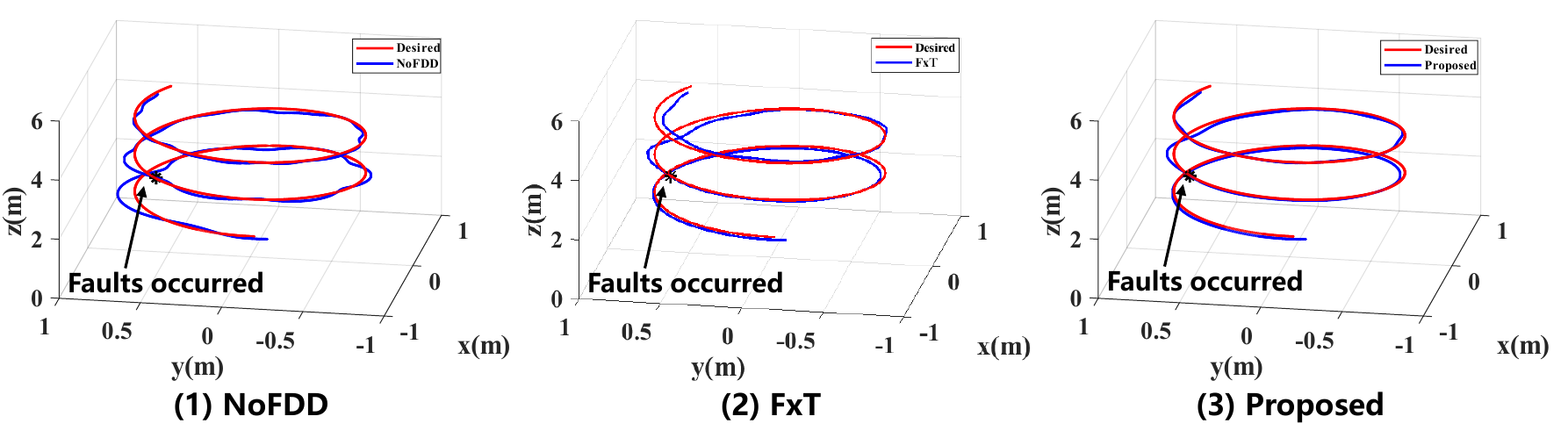}
	\caption{Comparison of position tracking performance.}
	\label{fig_4}
\end{figure*}

\subsection{Assessment}
The acutator faults estimation comparison between the proposed fault separation scheme and FxT are shown in Fig. \ref{fig_3}. Both the selected FDD schemes can achieve the robustness to the LoE actuator faults. Moreover, the effect of actuator faults and unknown input delay is separated and estimated because of the designed auxiliary control signal. Therefore, the estimation result of the proposed scheme is more stable and smooth than that of FxT due to the unknown dynamic of input delay is isolated. The standard deviation (STD) has been reduced by 46.3\% in comparison of FxT. In other words, the FDD adopting the proposed scheme can achieve a lower false alarm rate compared to the observer-based FDD scheme.

\begin{figure}[!t]
	\centering
	\includegraphics[width=3.5in]{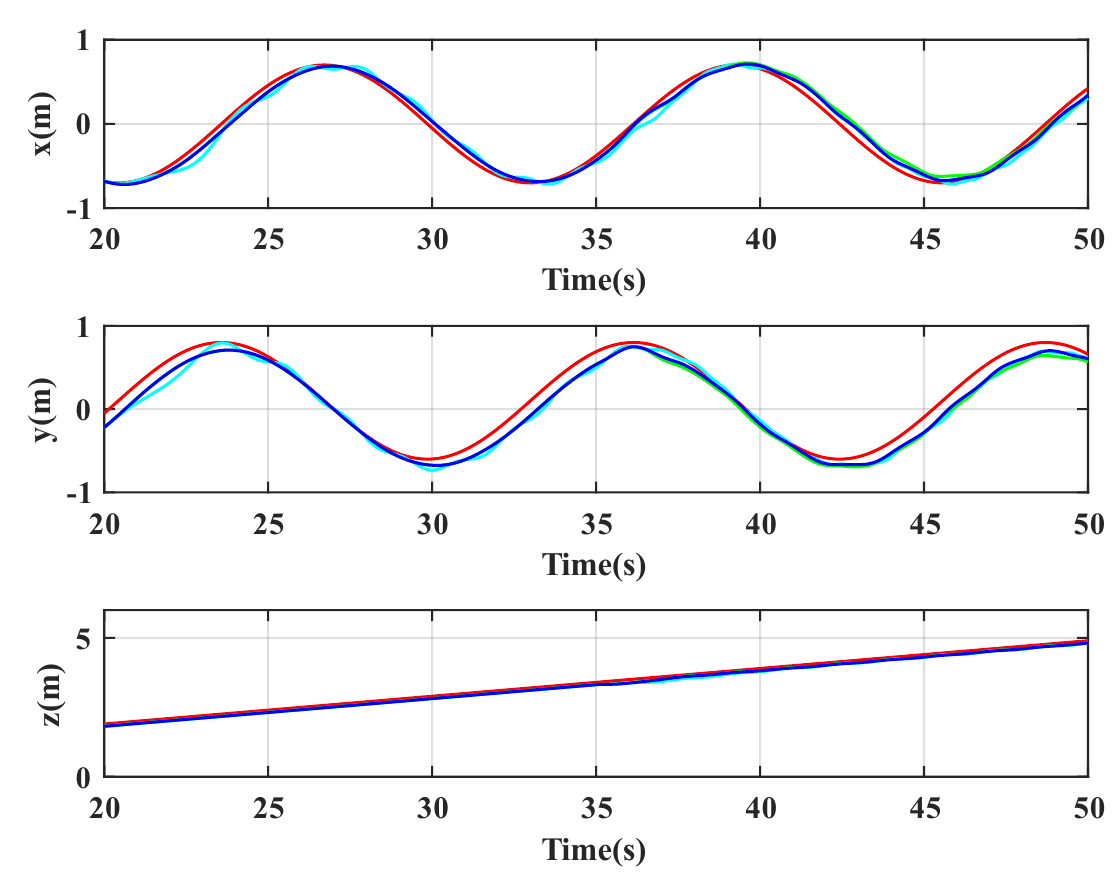}
	\caption{Time response of position tracking performance.}
	\label{fig_4.5}
\end{figure}

\begin{figure}[!t]
	\centering
	\includegraphics[width=3.5in]{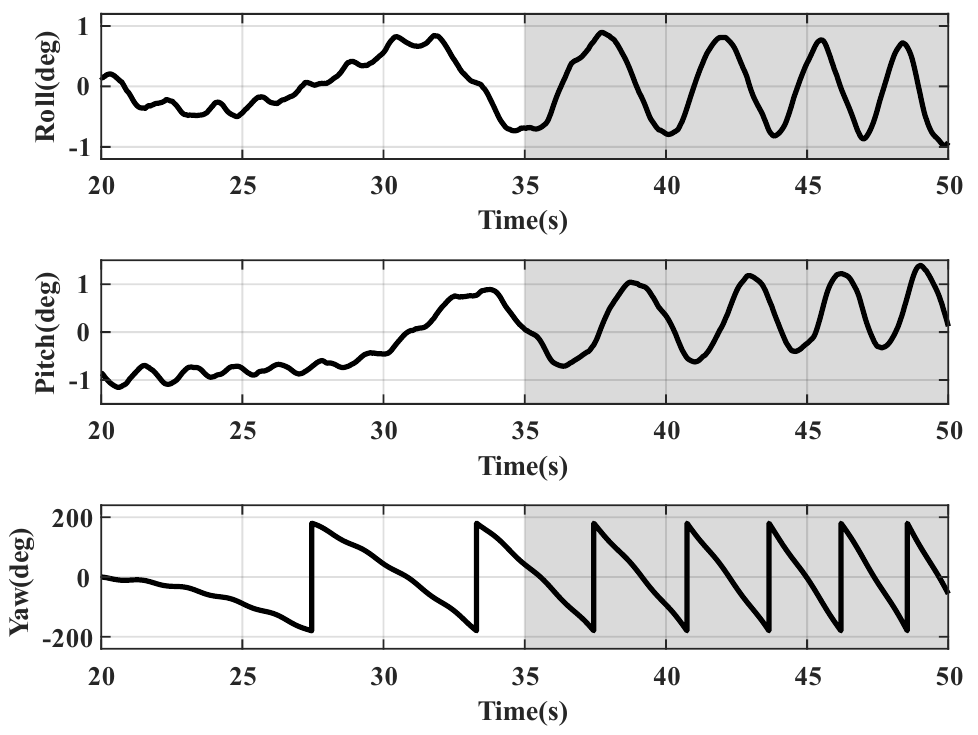}
	\caption{Attitude tracking performance: the gray area denotes the scenarios of LoE actuator fault occurrence.}
	\label{fig_5.5}
\end{figure}

The compared tracking performance among the proposed fault saperation scheme, FxT, and NoFDD is demonstrated in Fig. \ref{fig_4}.  As can be seen from Fig. \ref{fig_4}, The proposed fault saperation scheme achieves effective tracking of the desired trajectory, even if the yaw channel is sacrificed. Moreover, compared with FxT and NoFDD, the proposed fault saperation scheme has better tracking mean absolute error (MAE). Quantitatively, the tracking error of the proposed fault saperation scheme is 0.12$m$, while the errors of FxT and NoFDD is 0.26$m$ and 0.38$m$. The time responses of position and attitude are depicted in Figs. \ref{fig_4.5}-\ref{fig_5.5}. As a result, the proposed scheme is capable of not only ensuring the accuracy of FDD, but also preserving the tracking performance.

The thrust force calculated by the control allocation module are displayed in Fig. \ref{fig_6.5}. When the designed auxiliary control signal is injected, the control allocation module instantaneously adjusts the control inputs of other rotors to maintain the tracking performance as much as possible. As can be illustrated from Fig. \ref{fig_6.5}, owing to the injection of the designed auxiliary control signal, the command signal of M1 is sinusoidal. Meanwhile, the tracking performance of system state is guaranteed by sacrificing the controllability of yaw channel, as shown in Figs. \ref{fig_4}-\ref{fig_5}. Hence, from Figs. \ref{fig_3}-\ref{fig_6}, the simulation results illustrate the effectiveness of the proposed scheme. 
\begin{figure}[!t]
	\centering
	\includegraphics[width=3.5in]{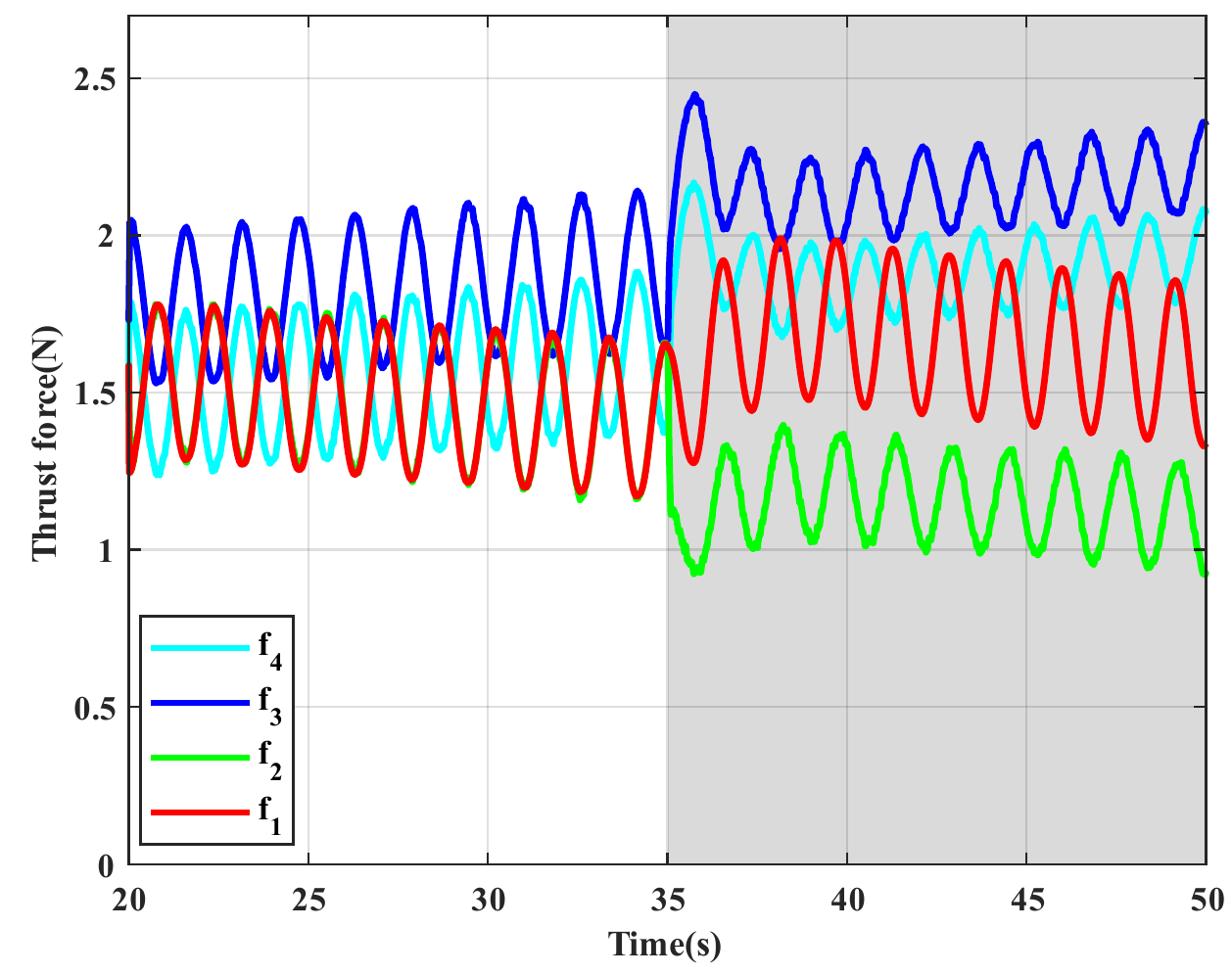}
	\caption{Thrust force generated by the proposed fault saperation and estimation scheme: the gray area denotes the scenarios of LoE actuator fault occurrence.}
	\label{fig_6.5}
\end{figure}

Moreover, a thousand Monte Carlo simulations based on the previous simulation are carried out. The indices are defined to quantify the experimental results. The mean estimation error (MEE) is defined as $\frac{1}{M}\sum\nolimits_{i=1}^{M}{\left\| {{\xi }_{\omega ,i}}-{{{\hat{\xi }}}_{\omega ,i}} \right\|}$, where $M$ is the length of data set. The mean tracking error (MTE) is defined as $\frac{1}{M}\sum\nolimits_{i=1}^{M}{\left\| {{p}_{d,i}}-{{p}_{i}} \right\|}$. The misdiagnosis rate (MR) is defined as $\frac{{{N}_{f}}}{{{N}_{s}}}$, where $N_s$ represents the number of experiments and $N_f$ denotes the number of misdiagnosis. 

The quantitative comparison is listed in Table \ref{table_4}. The proposed method outperforms either the anti-disturbance fault diagnosis observer and fault diagnosis framework. Especially for the MR index, compared with the other two methods, the MR of the proposed method is increased by $37.84\%$ and $70.27\%$ respectively. This is due to that the proposed excitation operator achieves high-precision by injecting the designed excitation operator. Moreover, the effect of LoE can be significantly attenuated by the designed safety controller.

\begin{table}[!t]
	\renewcommand{\arraystretch}{1.3}
	\caption{Monte Carlo Simulation Results} 
	\label{table_4}
	\centering
	\begin{tabular}{llll}
		\toprule
		Method & MEE & MTE(m) & MR($\%$) \\
		\midrule
		Ref. \cite{cao2017anti}   &$0.27$& $0.062$ &$5.1$\\
		\rowcolor{mygray} Ref. \cite{jia2022novel}   &$0.49$& $0.156$ &$6.3$\\
		Proposed &$0.18$& $0.053$ &$3.7$\\
		\bottomrule
	\end{tabular}
\end{table} 

\subsection{Experimental Results}
In order to exemplify the proposed excitation operator based FDD, comparative experimental tests are conducted. The overall architecture of the experimental system and arrangement is shown in Fig. \ref{fig_5}, which includes communication, navigation, and control modules. The positon of a quadrotor UAV can be obtained by a motion capture system.  The attitude of a quadrotor UAV can be measured with IMU and a motion capture system. The communication between the ground station and the quadrotor is transmitted by a UWB module.  An Lithium-Polymer battery, dual-blade propellers, and four customized motors are exploited to provide the corresponding thrusts. The load uncertainty of the flight experiment arises from the installation error of battery. The LoE fault and actuator aging are injected by software.

\begin{figure*}
	\centering
	\includegraphics[width=0.7\linewidth]{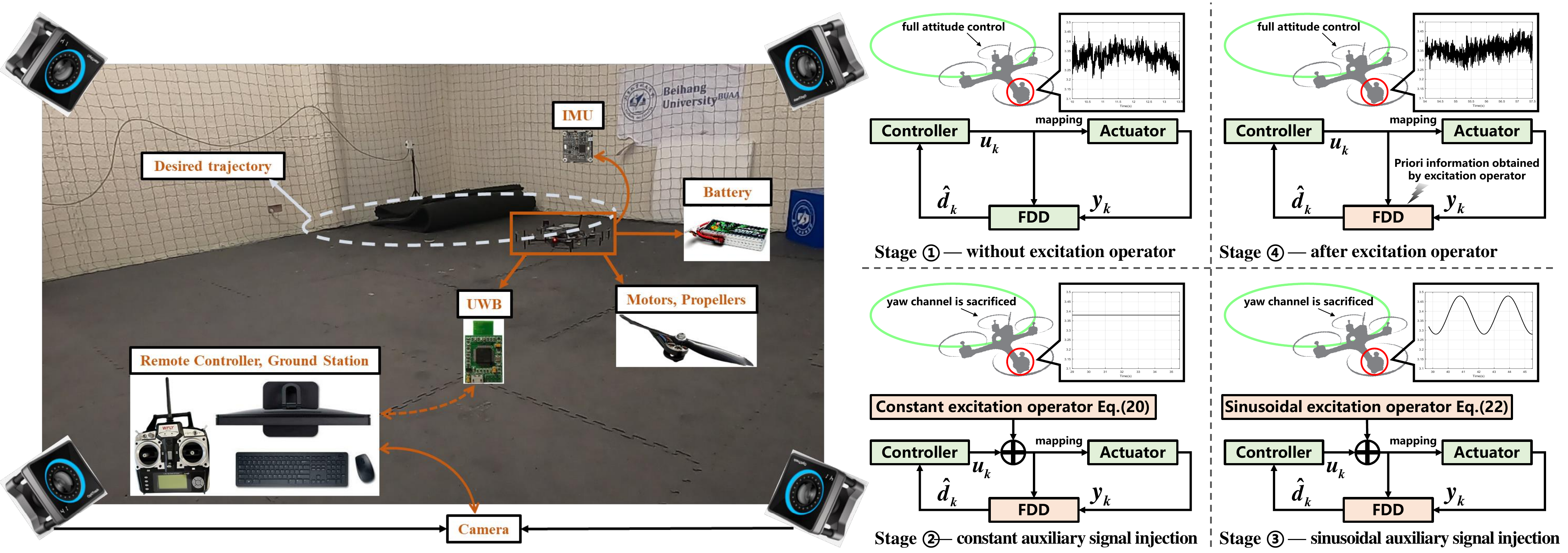}
	\caption{ \textcolor{black}{The overall framework and arrangement of the experimental test.}  }
	\label{fig_5}
\end{figure*}


The result of LoE degree estimated by the proposed FDD is illustrated in Fig. \ref{fig_6}. Due to the injection of the designed excitation operator, the LoE fault and actuator aging are successfully separated. Moreover, the estimation error of FDD decreases obviously. The MEE of the proposed method is $5.26\%$, while the anti-disturbance fault diagnosis observer is $8.74\%$. The root cause of the MEE is the measurement noise of the quadrotor UAV. Hence, from Fig. \ref{fig_6}, the proposed method can achieve superior fault estimation accuracy than that of the compared scheme.

\begin{figure}
	\centering
	\includegraphics[width=0.5\linewidth]{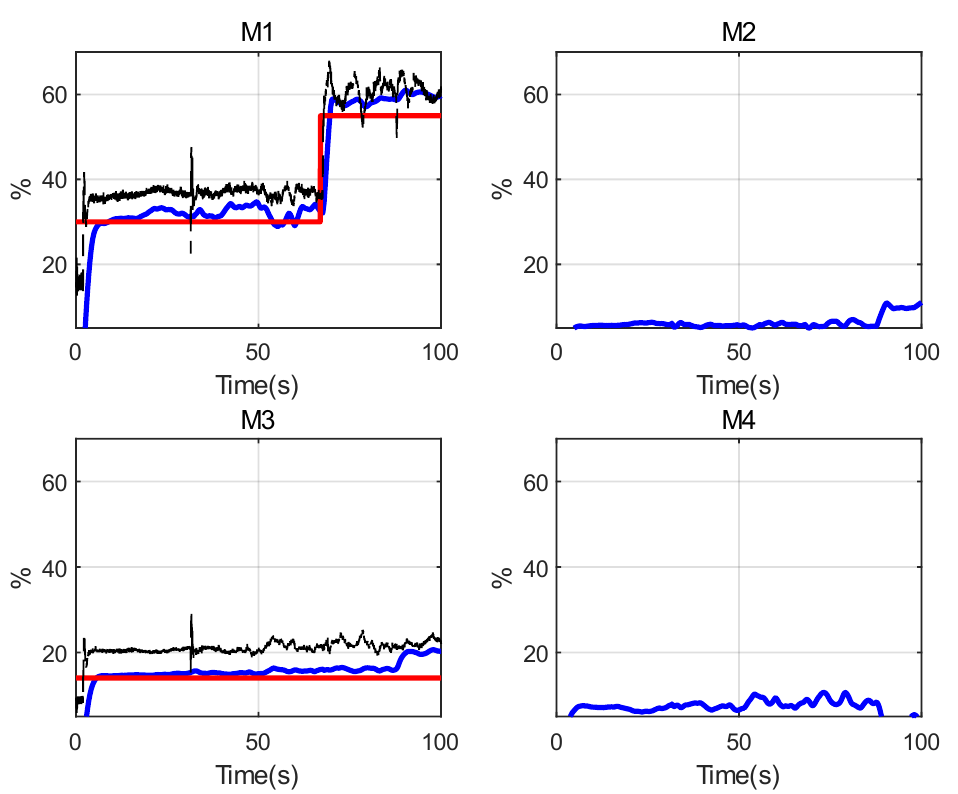}
	\caption{ \textcolor{black}{The estimation result of FDD.}  }
	\label{fig_6}
\end{figure}

\begin{figure}
	\centering
	\includegraphics[width=0.7\linewidth]{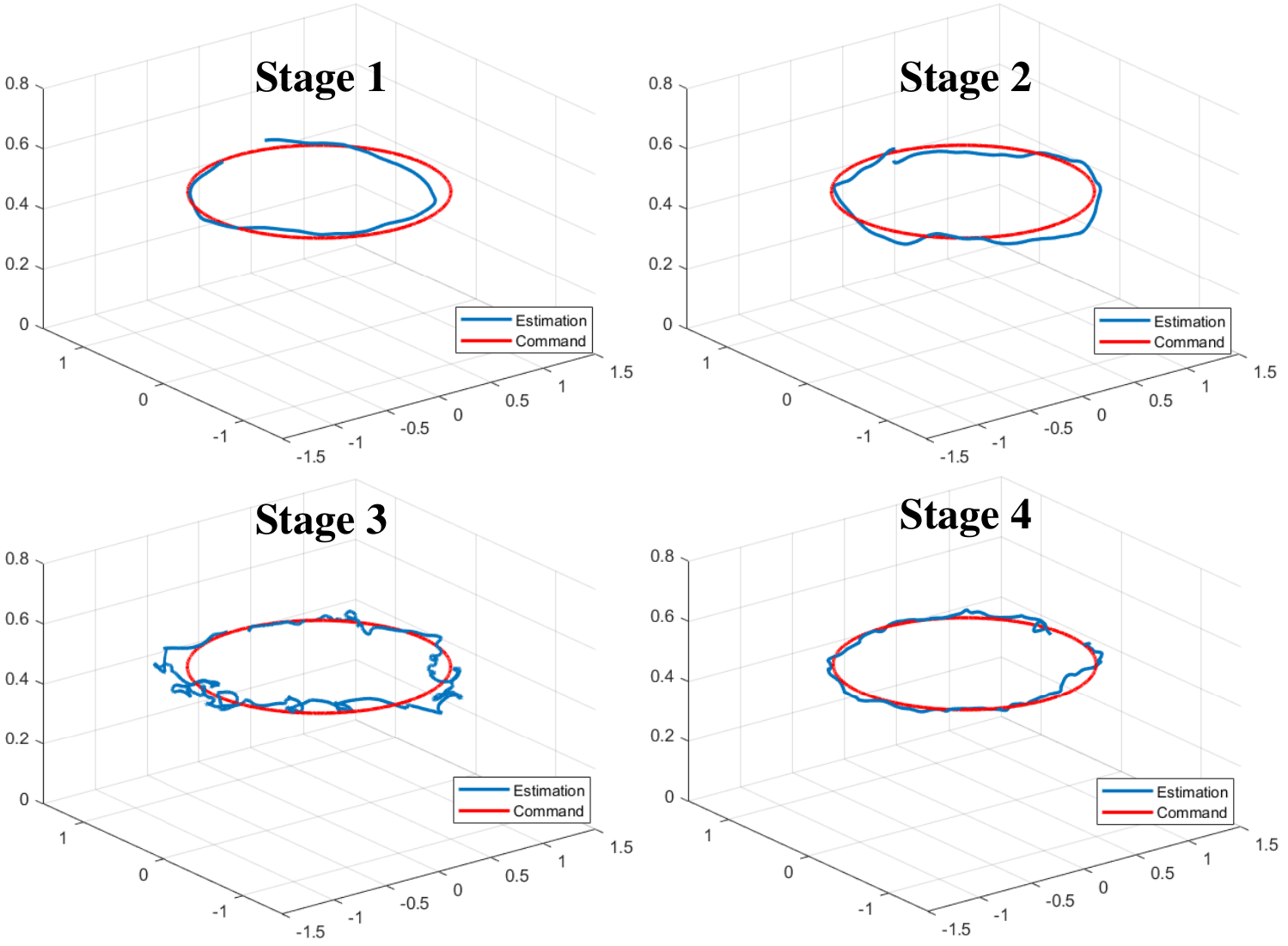}
	\caption{ \textcolor{black}{The position tracking performance}  }
	\label{fig_7}
\end{figure}

\begin{figure}
	\centering
	\includegraphics[width=0.7\linewidth]{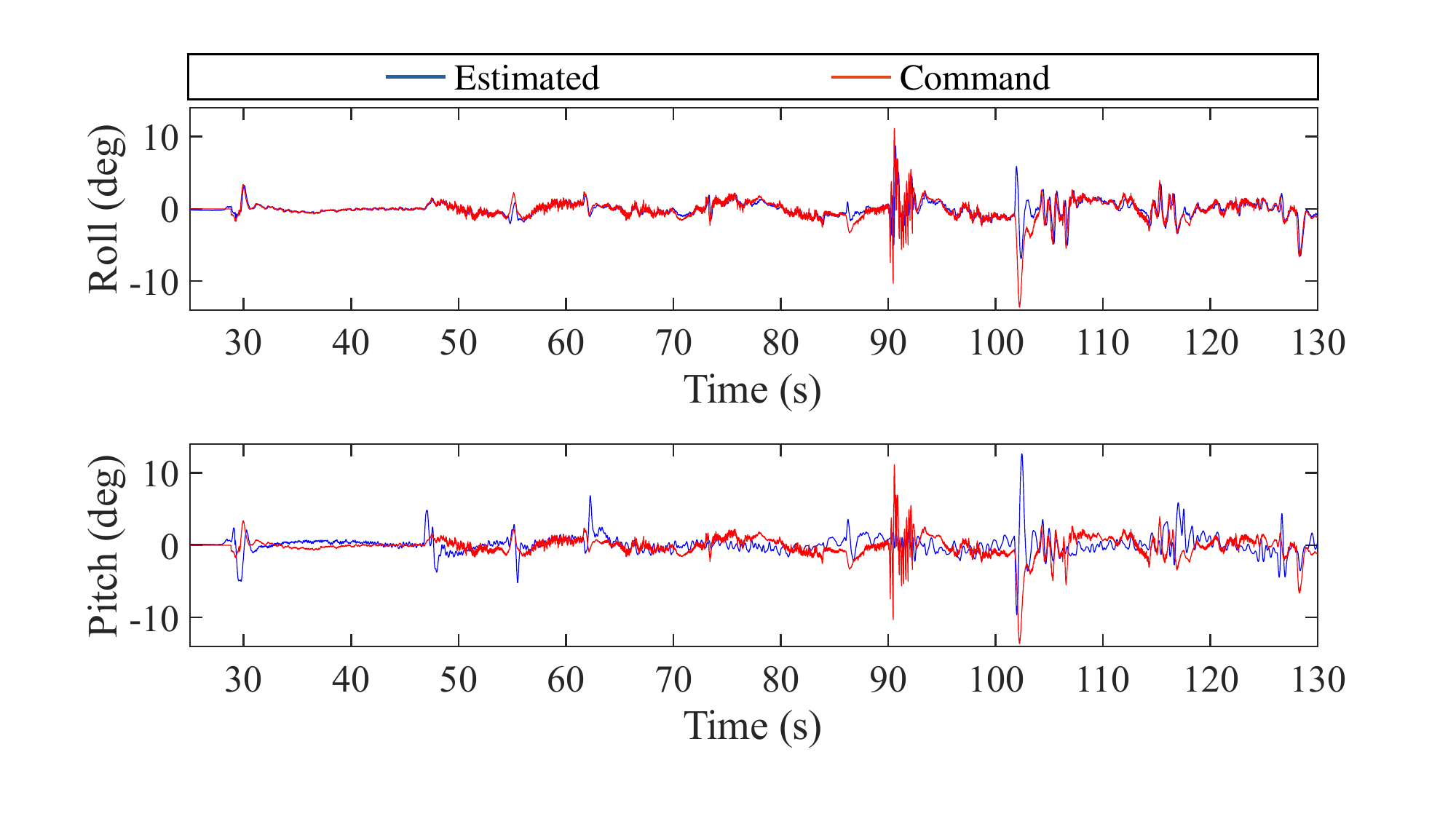}
	\caption{ \textcolor{black}{The attitude tracking performance}  }
	\label{fig_8}
\end{figure}

\begin{figure}
	\centering
	\includegraphics[width=0.7\linewidth]{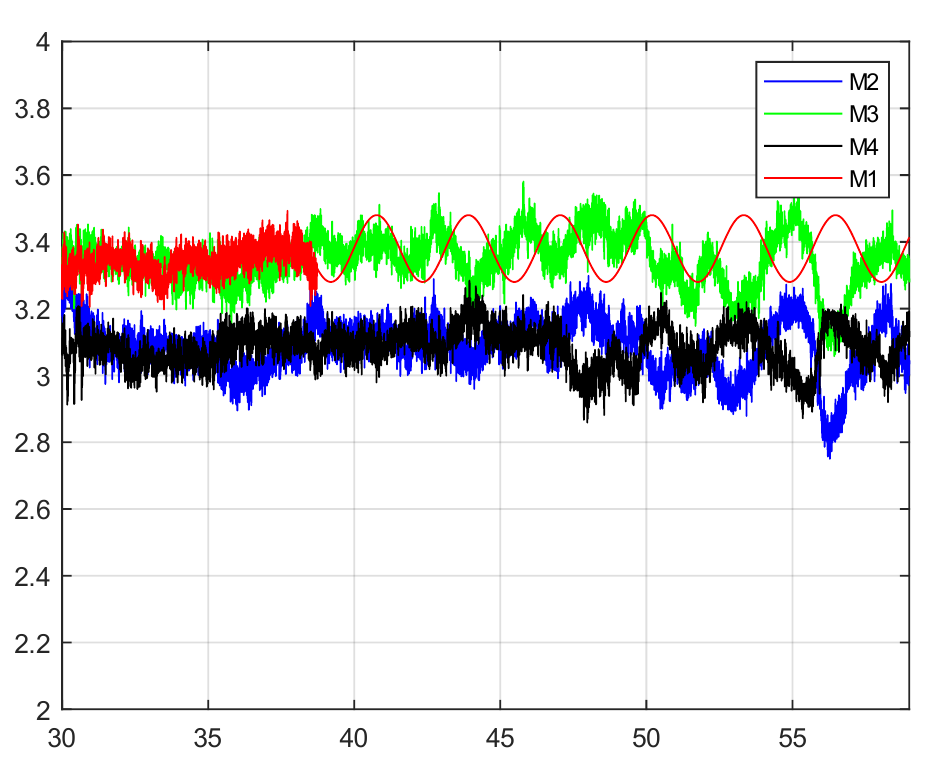}
	\caption{ \textcolor{black}{The thrust force commands.}  }
	\label{fig_9}
\end{figure}

The 3D tracking trajectory at each stage is demonstrated in Fig. \ref{fig_7}, while the tracking responses of attitude are depicted in Fig. \ref{fig_8}. Stages 1-4 represent tracking perfomance of post-fault quadrotor UAV without excitation operator, with constant auxiliary signal injection, with sinusoidal auxiliary signal injection, and after excitation operator injection, respectively. As can be seen from Fig. \ref{fig_7}, by resorting to the fault separation developed in this study, active actions for actuator health status can be taken. The MTE of the proposed safety control at each stage is $0.14m$, $0.19m$, $0.31m$, and $0.08m$, respectively. The tracking error of UAV increases during the injection of excitation operator. However, the tracking accuracy after fault diagnosis is improved than that before the injection of excitation operator.

\begin{table}[!t]
	\renewcommand{\arraystretch}{1.3}
	\caption{The Quantitative Results of Flight Tests} 
	\label{table_3_3}
	\centering
	\begin{tabular}{lllll}
		\toprule
		& Mean(m) & STD & Max(m) &MEE(\%) \\
		\midrule
		Proposed & $0.28$ & $0.0418$ & $0.23$ & $9.2$\\
		\rowcolor{mygray} Ref. \cite{cao2017anti} & $0.32$ & $0.0591$ & $0.29$ & $15.26$ \\
		Ref. \cite{jia2022novel} &  $0.61$ &  $0.0975$ & $0.35$ & $19.66$\\
		\bottomrule
	\end{tabular}
\end{table}

The thrust force commands generated by the proposed FDD scheme are presented in Fig. \ref{fig_9}. When the proposed FDD scheme is activated, the corresponding safety control strategy immediately adjusts the control inputs of the other rotors to maintain stability and tracking performance as much as possible. As shown in Figure \ref{fig_9}, the injection of the designed excitation operator results in a sinusoidal command signal for rotor M1. However, the tracking performance of the system state is guaranteed by sacrificing the controllability of the yaw channel, as demonstrated in Figs. \ref{fig_7} and \ref{fig_8}. Therefore, the experimental results presented in Figs. \ref{fig_7}-\ref{fig_9} confirm the applicability of the proposed scheme.


\section{Conclusion}
In this paper, an fault separation scheme is proposed. Similar to active fault diagnosis, an excitation operator is injected to achieve the higher precision estimation of actuator fault and load uncertainty. Moreover, the physical constraint and tracking performance are taken into consideration in the design of excitation operator. The controllability of yaw channel is sacrificed in order to ensure the 3-D trajectory tracking capability when the excitation operator is injected. Meanwhile, the safety controller is designed to compensate the mismatched disturbance caused by uncontrolled yaw. Comparative simulation and flight experiments have illustrated the effectiveness of the proposed active fault separation scheme where tracking performance can be guaranteed.

\ifCLASSOPTIONcaptionsoff
  \newpage
\fi



%

 


\bibliographystyle{IEEEtran}
\bibliography{IEEEabrv,mylib}

%







\end{document}